\documentclass[a4paper, 10pt, onecolumn]{article}
\usepackage[margin=0.55in]{geometry}

\usepackage[utf8]{inputenc}
\usepackage[T1]{fontenc}
\usepackage{color}
\usepackage{siunitx}
\usepackage{amsmath,verbatim,amssymb,amsfonts,amscd,graphicx,bm}
\usepackage{enumerate}
\allowdisplaybreaks

\DeclareMathOperator{\trace}{Tr}

\newtheorem{theorem}{Theorem}

\newtheorem{assumption}{Assumption}
\newtheorem{remark}{Remark}
\newtheorem{claim}{Claim}
\newtheorem{proof}{Proof}

\newcommand{\rhoa}{\bar{\rho}_A}
\newcommand{\rhos}{\rho_s}
\newcommand{\rhoak}{\bar{\rho}_A^{(k)}}

\newcommand{\rhoalessk}{\bar{\rho}_A^{[k]}}
\newcommand{\rhoakk}{\bar{\rho}_A^{(k,k^\prime)}}
\newcommand{\rhoalesskk}{\bar{\rho}_A^{[k,k^\prime]}}
\newcommand{\hc}{\textit{herm. conj.}}
\newcommand{\ket}[1]{ \left| #1 \right> }
\newcommand{\bra}[1]{ \left< #1 \right| }
\newcommand{\braket}[2]{\left< #1 {|} #2 \right> }

\title{\LARGE \bf
Adiabatic elimination for multi-partite open quantum systems with non-trivial zero-order dynamics
}

\author{Paolo Forni$^{1}$, Alain Sarlette$^{2}$, 
 Thibault Capelle$^{3}$, Emmanuel Flurin$^{3}$, Samuel Deléglise$^{3}$, \\ and Pierre Rouchon$^{1}$
 \thanks{$^{1}$Centre Automatique et Systèmes, Mines-ParisTech, PSL Research
 University. 60 Bd Saint-Michel, 75006 Paris, France; and INRIA Paris, 2 rue Simone Iff, 75012 Paris, France.
 {\tt\small paolo.forni@inria.fr,  pierre.rouchon@mines-paristech.fr}. }
 \thanks{$^{2}$INRIA Paris, 2 rue Simone Iff, 75012 Paris, France; and Ghent
 University / Data Science Lab, Technologiepark 914, 9052 Zwijnaarde,
 Belgium. {\tt\small alain.sarlette@inria.fr}. }%
 \thanks{$^{3}$Laboratoire Kastler Brossel, ENS-PSL, CNRS, Sorbonne Université et Collège de France, Paris, France.
 {\tt\small thibault.capelle@lkb.upmc.fr}, {\tt\small emmanuel.flurin@lkb.upmc.fr}, {\tt\small samuel.deleglise@lkb.upmc.fr}. }%
 }

\begin{document}

\maketitle
\thispagestyle{empty}
\pagestyle{empty}

\begin{abstract}
We provide model reduction formulas for open quantum systems consisting of a target component 
which weakly interacts with a strongly dissipative environment. The time-scale separation between 
the uncoupled dynamics and the interaction allows to employ tools from center manifold theory and 
geometric singular perturbation theory to eliminate the variables associated to the environment 
(adiabatic elimination) with high-order accuracy. An important specificity is to preserve the
quantum structure: reduced dynamics in (positive) Lindblad form and coordinate mappings in Kraus
form. We provide formulas of the reduced dynamics. The main contributions of this paper are (i)
to show how the decomposition of the environment into $K$ components enables its efficient treatment,
avoiding the quantum curse of dimension; and (ii) to extend the results to the case where the target 
component is subject to Hamiltonian evolution at the fast time-scale. We apply our theory to a 
microwave superconducting quantum resonator subject to material losses, and we show
that our
reduced-order model 
can explain the transmission spectrum observed in a recent pump probe experiment.
\end{abstract}

\section{Introduction}
\label{sec:intro}

The evolution of a quantum system interacting with an environment is rigorously described by a Schr\"{o}dinger
equation on the joint Hilbert space. However, the complexity of the environment hampers the study of the system as a whole and one often resorts to the Born-Markov approximation to obtain a Lindblad master equation \cite{breuer2002theory} describing the target system alone, and the environment's effect summarized by dissipation or ``decoherence'' operators. Similarly, when a quantum system consists of several interacting components, e.g. a main computing subsystem coupled to an ancillary subsystem expressing a measurement device, one often seeks to analyze a dynamical equation for the main subsystem alone, approximately including the effect of the ancillary subsystem.
In this perspective, model reduction methods come to aid to the physicists interested in gaining better physical insights, in running simplified numerical simulations, and in designing the dynamics of a target subsystem by smartly engineering its interaction with other subsystems, as in the case of reservoir engineering \cite{reservoirengineering1996}.

A classical approach to model reduction for quantum systems makes use of the time-scale separation between
a slow subsystem of interest and the fast auxiliary subsystems coupled to it, and eliminates
the fast variables in a procedure denominated as \emph{adiabatic elimination}. 
In closed quantum systems -- where the evolution stays unitary under Hamiltonian dynamics --
adiabatic elimination is performed by means of standard perturbation theory techniques \cite{sakurai2017modern}.
In contrast, the treatment of open quantum systems -- including decoherence under Lindbladian dynamics -- is more involved. In the literature, adiabatic elimination in the latter case has been 
addressed for specific examples separately: lambda systems up to second-order \cite{adiabelimlambdasys},
a specific atom-optics example \cite{atkins2003}, systems where excited states decay toward $n$ ground states
\cite{mirrahimirouchon2009,reitersorensen2012}, systems with Gaussian dynamics and subject to continuous
measurement \cite{adiabelimgauss}.

However, general approaches to adiabatic elimination of Lindblad systems -- and maintaining the positivity-preserving quantum structure, beyond a standard linear systems treatment via singular perturbation theory -- have attracted much less attention. In \cite{kesslerSW}, Kessler has developed a generalization of the Schrieffer-Wolff formalism; in \cite{Gough2007,BOUTEN20083123}, the authors address quantum stochastic differential equations in the limit where the speed of the fast system goes to infinity. A geometric approach to adiabatic elimination has been introduced by \cite{azouit2016cdc,azouitQST2017}, where the authors explore an asymptotic expansion of the reduced dynamics by a careful application of center manifold techniques \cite{carr2012applications} and geometric singular perturbation theory \cite{FENICHEL197953}. In order to succesfully retain the physical interpretation, the reduced dynamics is expressed by Lindblad equations and is mapped to the original dynamics via a trace-preserving completely-positive (CPTP) map, also called Krauss map.

The present work builds upon the geometric approach of \cite{azouit2016cdc,azouitQST2017} and brings forward two novel features. First, unlike in \cite{azouit2016cdc,azouitQST2017} where the target system was assumed to be static in the ideal case, we here develop formulas for the case where the target system undergoes non-trivial fast Hamiltonian dynamics, when uncoupled from the environment.
This appears in all practical situations where the target system is detuned from the reference frame, e.g. when the target system undergoes (in this paper constant) drives to implement quantum operations.
Second, we consider environments that consist not of a single bulk system, but which can be decomposed into a not-necessarily-finite number of fast dissipative subsystems. Such situations often appear in practice when the target quantum system is corrupted by various imperfection sources \cite{TLSpaper}.
We show how to take advantage of this decomposition towards more efficient model reduction computations.
Indeed, the first-order approximation amounts to the sum of the contributions of each fast dissipative
subsystem, and the same result holds for the second-order approximation under specific commutation
properties of the operators involved in the computation.
This is a substantial gain because the difficult operations involve
inversion of the Lindbladian superoperator precisely over the environment dimension.
The proposed theory is applied to a model of a microwave superconducting resonator subject to dielectric
losses due to a bath of many two-level-systems. We show how a reduced model resulting from our theory allows to 
explain the non-trivial transmission spectrum observed in a pump probe experiment.

The outline of the paper is as follows. Setting and main assumptions are introduced in Section \ref{sec:setting}.
Section \ref{sec:manyfast} provides our main results with the formulas of our adiabatic elimination for the case of
many fast dissipative subsystems weakly coupled to the target one.
Section \ref{sec:application} contains the application and comparison
to experimental data. We conclude the paper with few final remarks. Proof and computation details are given in
appendix.

\section{Setting}
\label{sec:setting}

\subsection{K-partite systems with non-trivial zero-order dynamics}
\label{ssec:settingbipartite}

Open quantum systems are typically described by differential equations evolving on the manifold $\mathcal{M}$ of
density operators $\rho$, namely the set of all
linear Hermitian nonnegative operators from a Hilbert space $\mathcal{H}$ to itself, whose trace equals one.
The evolution of an open quantum system is then described by the Lindblad master equation
\cite{breuer2002theory}:
\begin{equation}
 \frac{d \rho}{d t} = \mathcal{L}(\rho) = -i \left[ \bm{H},\rho \right] +
 \sum_{\mu} { \mathcal{D}_{\bm{L}_\mu}(\rho) } \nonumber , 
\end{equation}
where each $\bm{L}_\mu$ is a ``decoherence'' operator on $\mathcal{H}$, $\bm{H}$ is a Hermitian ``Hamiltonian'' operator on $\mathcal{H}$,
and $\mathcal{D}$ is a superoperator defined by:
\begin{equation}
\mathcal{D}_{\bm{L}_\mu} (\rho) := \bm{L}_\mu \rho \bm{L}^\dag_\mu - \frac{1}{2} \bm{L}^\dag_\mu  \bm{L}_\mu  \rho -\frac{1}{2} \rho  \bm{L}^\dag_\mu  \bm{L}_\mu   \nonumber .
\end{equation}

In this paper, we consider the composite Hilbert space $\mathcal{H} := \mathcal{H}_A \otimes \mathcal{H}_B$ of a target quantum system on $\mathcal{H}_B$ and its environment on $\mathcal{H}_A$. The dynamics on $\mathcal{H}$ satisfies a time scale separation:
\begin{equation}
 \frac{d \rho}{dt} =   \mathcal{L}_A(\rho) + \varepsilon \mathcal{L}_{int}(\rho)
 +\varepsilon \mathcal{L}_B(\rho) + (-i) [ \tilde{\bm{H}}_B, \, \rho ]  \label{eq:systemsinglefast} ,
\end{equation}
where $\varepsilon$ is a small positive parameter; $\mathcal{L}_A$ and $\mathcal{L}_B$ are Lindbladian super-operators acting exclusively on $\mathcal{H}_A$ and $\mathcal{H}_B$ respectively; $\mathcal{L}_{int}$ is a Lindbladian superoperator which captures the interaction between $\mathcal{H}_A$ and $\mathcal{H}_B$. Here we assume that this interaction is Hamiltonian and expressed as:
\begin{equation}
 \mathcal{L}_{int}(\rho) := -i \left[ \bm{A} \otimes \bm{B}^\dag + \bm{A}^\dag \otimes \bm{B} , \,\, \rho \right] \nonumber ,
\end{equation}
where $\bm{A}$ and $\bm{B}$ respectively are non-necessarily-Hermitian operators acting on 
$\mathcal{H}_A$ and $\mathcal{H}_B$ only.
The resonant interaction from of $\mathcal{L}_{int}$ models a wide range of applications; general interactions
will be addressed by future works.
Finally, $\tilde{\bm{H}}_B$ is a Hamiltonian operator on $\mathcal{H}_B$, thus expressing fast unitary dynamics on the target system; its presence is the first novelty in our paper. For a set of interesting situations, the dynamics of typical quantum systems can be expressed in a rotating frame where the term $\tilde{\bm{H}}_B$ would vanish. However, several reasons can justify to keep this term. For instance, in many significant situations the vanishing of $\tilde{\bm{H}}_B$ is not rigorous and involves an additional treatment of appearing fast time-varying parameters in the equation via averaging theory; or, $\tilde{\bm{H}}_B$ can be a term of particular interest like a field to be measured with the quantum device or an actuation towards applying some operation on the target system.

As a second novelty, we consider a generalized setting where $\mathcal{H}_A = \bigotimes_{k}\mathcal{H}_A^{(k)} $ is composed of a non-necessarily-finite number of Hilbert spaces $\mathcal{H}_A^{(k)}$. Each subsystem
on $\mathcal{H}_A^{(k)}$ is strictly dissipative and interacts with $\mathcal{H}_B$ only. Then, system \eqref{eq:systemsinglefast} reads as:
\begin{equation}
 \frac{d \rho}{dt} = \sum_{k}{ \Big( \mathcal{L}_A^{(k)}(\rho) + \varepsilon \mathcal{L}_{int}^{(k)}(\rho) \Big) }
 + \varepsilon \mathcal{L}_B(\rho) + (-i) \left[ \tilde{\bm{H}}_B, \, \rho \right] \label{eq:systemmanyfast}
\end{equation}
where $\mathcal{L}_A^{(k)}$ acts on
$\mathcal{H}_A^{(k)}$ only and where
\begin{align}
  \mathcal{L}_{int}^{(k)}(\rho) := & -i \left[ \bm{A}^{(k)} \otimes \bm{B}^\dag + \bm{A}^{(k)\,\dag} \otimes \bm{B} , \,\, \rho \right] \nonumber ,
\end{align}
captures the Hamiltonian interaction between $\mathcal{H}_A^{(k)}$ and $\mathcal{H}_B$, with
$\bm{A}^{(k)}$ non-necessarily-Hermitian operators acting on $\mathcal{H}_A^{(k)}$ only.
The interaction is here restricted to the case of the same operator $\bm{B}$ for each subsystem $k$.
While the general case will be the subject of future research, having the same operator $\bm{B}$ for each interaction
still models a wide range of applications.

For $\varepsilon=0$, the system is uncoupled and the solution trajectories stay separable for all times, namely for $\rho(0) = \bigotimes_{k}\rho_A^{(k)}(0) \otimes \rho_B(0)$ we have $\rho(t) = \bigotimes_{k}\rho_A^{(k)}(t) \otimes \rho_B(t)$ for all times, with each factor in the product following its independent dynamics. To apply adiabatic elimination, we assume that each part of the environment is highly dissipative and relaxes fast to a unique steady state, i.e.: for any initial state $\rho_0$ on $\mathcal{H}_A \otimes \mathcal{H}_B$,
the solution of the uncoupled system $\varepsilon=0$ converges to $\bigotimes_{k}{\bar{\rho}_A^{(k)}} \otimes \rho_B(t)$ where, for each $k$, $\bar{\rho}_A^{(k)}$ is the unique solution of $\mathcal{L}_A^{(k)}\left( \bar{\rho}_A^{(k)} \right) =0$; and $\rho_B(t)$ satisfies 
$\dot{\rho}_B = -i [\tilde{\bm{H}}_B,\rho_B]$ with $\rho_B(0) = \trace_A(\rho_0)$.
For ease of presentation, we will also denote 
$\bar{\rho}_A := \bigotimes_{k}{\bar{\rho}_A^{(k)}}$.

\subsection{Asymptotic expansion}
\label{ssec:asymptoticexpansion} 

Both in the bi-partite and the $K$-partite case, for the uncoupled system $\varepsilon=0$, there exists an asymptotically stable center manifold $\mathcal{M}_0$ of same dimension as $\mathcal{H}_B$, on which the dynamics  have imaginary eigenvalues. It thus follows from Fenichel's Invariant Manifold Theorem \cite{FENICHEL197953} that,
for small enough $\varepsilon>0$, there exists an invariant and attractive manifold $\mathcal{M}_\varepsilon$
which has the same dimension as $\mathcal{M}_0$ and which is $\varepsilon-$close to it.
Furthermore, by virtue of linearity and Carr's result \cite{carr2012applications},
$\mathcal{M}_\varepsilon$ is a vector subspace and its approximation can be computed up to
arbitrary precision.
The quantum particularity, as explained in \cite{azouitQST2017}, is that such approximation should retain a physical interpretation by preserving the quantum structure:
(i) the mapping from the reduced space to the complete space is a mapping between density operators,
and $\mathcal{M}_\varepsilon$ can be parameterized by
$\mathcal{M}_\varepsilon := \left\{ \rho \in \mathcal{H} : \, \rho = \mathcal{K}(\rho_s), \, \rho_s \in \mathcal{H}_s \right\}$
for some Hilbert space $\mathcal{H}_s$ that has same dimension as $\mathcal{H}_B$, and where
$\mathcal{K}(\cdot )$ is a Kraus map\footnote{A Kraus map takes the form
$\rho = \mathcal{K}(\rho_s) := \sum_{\ell}{ \bm{M}_\ell \, \rho_s  \bm{M}_\ell^\dag }$
for some operators $\bm{M}_\ell$ in order to express any completely positive superoperator \cite{CHOI1975285}, and with $\sum_{\ell}{ \bm{M}_\ell \, \rho_s \bm{M}_\ell^\dag } = \bm{I}$ ensuring trace-preservation i.e. $\trace \left( \mathcal{K}(\rho_s) \right)= \trace(\rho_s) = 1$.
}; (ii) the reduced dynamics on $\mathcal{M}_\varepsilon$ are Lindbladian, i.e. $\dot{\rho_s}=\mathcal{L}_s(\rho_s)$ for some Lindbladian superoperator $\mathcal{L}_s$.

In other words, we aim to find a Kraus map $\rho = \mathcal{K}(\rho_s)$ and a Lindbladian $\mathcal{L}_s$ such that
the following invariance equation is satisfied for all $\varepsilon$ small enough and
for all $\rho_s$:
\begin{align}
 & \mathcal{L}_A(\mathcal{K}(\rho_s)) + \varepsilon \mathcal{L}_{int}(\mathcal{K}(\rho_s))
 +\varepsilon \mathcal{L}_B(\mathcal{K}(\rho_s)) \nonumber   \\
 & + (-i) [ \tilde{\bm{H}}_B, \, \mathcal{K}(\rho_s) ] \;\;\;\;  =
 \mathcal{K}(\mathcal{L}_{s}(\rho_s)) \label{eq:invariancecondition} .
\end{align}
By virtue of 
Carr's result \cite{carr2012applications},
we first parameterize both the Kraus map and the Lindbladian as infinite series:
\begin{equation}
 \mathcal{K}(\rho_s) := \sum_{h=0}^{+\infty}{\varepsilon^h \,\, \mathcal{K}_h(\rho_s)}  , \;\;\;\; 
 \mathcal{L}_s(\rho_s) := \sum_{h=0}^{+\infty}{\varepsilon^h \,\, \mathcal{L}_{s,h}(\rho_s)} \label{eq:param} ;
\end{equation}
then, by identifying the terms of the same order of $\varepsilon$ in the invariance equation (\ref{eq:invariancecondition}),
we obtain an invariance relation at all orders $h$. At zero-order, we have:
\begin{equation}
 \mathcal{L}_A\left(\mathcal{K}_0(\rho_s)\right) + 
 (-i) \left[ \tilde{\bm{H}}_B, \mathcal{K}_0(\rho_s) \right]  = \mathcal{K}_0 \left( \mathcal{L}_{s,0}(\rho_s) \right) \label{eq:zeroorderinvariance} .
\end{equation}
Similarly, the first-order invariance condition reads as:
{\small \begin{align}
& \mathcal{L}_A\left(\mathcal{K}_1(\rho_s)\right) 
 + \mathcal{L}_{int}\left( \mathcal{K}_0(\rho_s) \right)
 + \mathcal{L}_B\left( \mathcal{K}_0(\rho_s) \right) \nonumber \\
& \;\;\;\; -i \left[ \tilde{\bm{H}}_B , \, \mathcal{K}_1(\rho_s) \right] = 
 \mathcal{K}_0\left( \mathcal{L}_{s,1}(\rho_s)  \right)
 +\mathcal{K}_1\left(\mathcal{L}_{s,0}(\rho_s)  \right)  \label{eq:firstorderinvariance} ,
\end{align}}

{\parskip = -2mm
\noindent
whereas the second-order invariance condition reads as:}
{\small \begin{align}
& \mathcal{L}_A\left( \mathcal{K}_2(\rho_s)  \right) 
 + \mathcal{L}_{int} \left( \mathcal{K}_1(\rho_s) \right)   \nonumber \\
 & \;\;\;\;\; + \mathcal{L}_B\left( \mathcal{K}_0(\rho_s) \right) 
 -i \left[ \tilde{\bm{H}}_B , \, \mathcal{K}_2(\rho_s) \right] \nonumber \\
 & = 
 \mathcal{K}_0\left( \mathcal{L}_{s,2}(\rho_s)  \right)
 +\mathcal{K}_1\left( \mathcal{L}_{s,1}(\rho_s)  \right)
 +\mathcal{K}_2\left(\mathcal{L}_{s,0}(\rho_s)  \right) .  \label{eq:secondorderinvariance}
\end{align}}
For more details about the asymptotic expansion approach to adiabatic elimination,
we refer the reader to \cite{azouitQST2017,azouit2016cdc}.

\section{Reduced-model formulas}
\label{sec:manyfast}

The aim of this Section is to provide explicit solutions to the zero-, first-, and second-order invariance equations (\ref{eq:zeroorderinvariance})-(\ref{eq:secondorderinvariance}) for the case of $K$-partite systems as introduced
in Section \ref{ssec:settingbipartite}, i.e. for model (\ref{eq:systemmanyfast}). 
We immediately observe that the zero-order (\ref{eq:zeroorderinvariance}) is naturally solved by setting:
\begin{equation}
\mathcal{L}_{s,0}(\rho_s)  : =  -i \left[ \tilde{\bm{H}}_B , \, \rho_s \right], \;\;\;\;
\mathcal{K}_0(\rho_s)  : = \left( \bigotimes_{k}{\bar{\rho}_A^{(k)}} \right) \otimes \rho_s \label{eq:zeroorderK} .
\end{equation}

At first order, let the Kraus map have the following structure inspired by \cite{azouitQST2017}:
\small \begin{align}
 & \mathcal{K}(\rho_s) = \mathcal{K}_0(\rho_s) + \varepsilon \mathcal{K}_1(\rho_s) :=  \nonumber \\
 & \;\; \left( \bm{I} -i \varepsilon \bm{M} \right) \left(\bar{\rho}_A \otimes \rho_s \right)
 \left( \bm{I} +i \varepsilon \bm{M}^\dag \right)  + \mathcal{O}\left( \varepsilon^2 \right) \label{eq:kmap10} ,
\end{align}
where $ \bm{M} := \sum_{k}  \bm{M}^{(k)}$, $\bm{M}^{(k)} :=   \bm{F}_1^{(k)} \otimes \bm{B}^\dag + \bm{F}_2^{(k)} \otimes \bm{B}$ for any $k$.
This would immediately imply that:
\begin{equation}
 \mathcal{K}_1(\rho_s) = -i \bm{M} \left(\bar{\rho}_A \otimes \rho_s \right) + i
 \left(\bar{\rho}_A \otimes \rho_s \right) \bm{M} \label{eq:K1shape00} .
\end{equation}


The following assumption will be instrumental in establishing our main results.
\begin{assumption}
\label{assumption:cbassumption}
 There exists $c_{B^\dag} \in \mathbb{R}$ such that:
\begin{equation}
\left[ \tilde{\bm{H}}_B, \bm{B}^\dag\right] = c_{B^\dag} \bm{B}^\dag  \label{eq:cBproperty}.
\end{equation} 
\end{assumption}
\begin{theorem}
\label{th:manyFastFirstOrder}
Consider model (\ref{eq:systemmanyfast}). Let Assumption \ref{assumption:cbassumption} hold.
Then, the first-order invariance equation (\ref{eq:firstorderinvariance}) is satisfied by
the Lindbladian $\mathcal{L}_{s,1}(\rho_s) = \mathcal{L}_B(\rho_s)$ and by a map $\mathcal{K}_1$
 of the form (\ref{eq:K1shape00}) 
 where, for each $k$, $\bm{F}_1^{(k)},\bm{F}_2^{(k)}$ respectively are the unique solutions of:
{\small \begin{subequations}
\begin{align}
 \mathcal{L}_A^{(k)}\left( \bm{F}_1^{(k)} \bar{\rho}_A^{(k)}  \right)  + \bm{A}^{(k)} \bar{\rho}_A^{(k)} -i \, c_{B^\dag}\, \bm{F}_1^{(k)} \bar{\rho}_A^{(k)} = & 0, \label{eq:singleFastFirstOrderSelectionOneMany1} \\
 \mathcal{L}_A^{(k)}\left( \bm{F}_2^{(k)} \bar{\rho}_A^{(k)}  \right)  + \bm{A}^{(k)\,\dag} \bar{\rho}_A^{(k)} + i \, c_{B^\dag}^\ast \, \bm{F}_2^{(k)} \bar{\rho}_A^{(k)} = & 0 \label{eq:singleFastFirstOrderSelectionOneMany2} .
\end{align} \label{eq:singleFastFirstOrderSelectionOneMany}
\end{subequations}}
Furthermore, $\mathcal{K}(\rho_s) = \mathcal{K}_0(\rho_s) + \varepsilon \mathcal{K}_1(\rho_s)$ is a CPTP map up to second-order terms.
\end{theorem}
\begin{proof}
see Appendix \ref{ssec:manyFastFirstOrder}.
\end{proof}
\begin{remark}
\label{remark:firstorder}
The first-order (\ref{eq:firstorderinvariance}) is also satisfied by
the Lindbladian $\mathcal{L}_{s,1}(\rho_s) = \mathcal{L}_B(\rho_s) -i \left[ \bm{H}_{s,1} , \rho_s \right]$
 with $\bm{H}_{s,1} := \sum_{k} \trace \left( \bm{A}^{(k)} \bar{\rho}_A^{(k)} \right)  \bm{B}^\dag + \trace \left( \bm{A}^{(k)\,\dag} \bar{\rho}_A^{(k)} \right)  \bm{B}$,
and by a map $\mathcal{K}_1$ of the form (\ref{eq:K1shape00}) 
 where $\bm{F}_1^{(k)},\bm{F}_2^{(k)}$ respectively are the unique solutions of:
\begin{subequations} 
\begin{align}
 \mathcal{L}_A^{(k)}\left( \bm{F}_1^{(k)} \bar{\rho}_A^{(k)}  \right)  + \mathcal{S}^{(k)}(  \bm{A}^{(k)}  \bar{\rho}_A^{(k)})    -i \, c_{B^\dag}\, \bm{F}_1^{(k)} \bar{\rho}_A^{(k)} = & 0, \label{eq:singleFastFirstOrderSelectionTwoMany1} \\
 \mathcal{L}_A^{(k)}\left( \bm{F}_2^{(k)} \bar{\rho}_A^{(k)}  \right)  + \mathcal{S}^{(k)}(  \bm{A}^{(k)\,\dag}  \bar{\rho}_A^{(k)}) + i \, c_{B^\dag}^\ast \, \bm{F}_2^{(k)} \bar{\rho}_A^{(k)} = & 0 \label{eq:singleFastFirstOrderSelectionTwoMany2} ,
\end{align} \label{eq:singleFastFirstOrderSelectionTwoMany} 
\end{subequations}

{\parskip = -2mm
\noindent
where, for an operator $Q$ acting on $\mathcal{H}_A^{(k)}$, notation $\mathcal{S}^{(k)}(Q)$
denotes $Q-\trace(Q)\bar{\rho}_A^{(k)}$.
Furthermore, $\mathcal{K}(\rho_s) = \mathcal{K}_0(\rho_s) + \varepsilon \mathcal{K}_1(\rho_s)$ is a CPTP map up to second-order terms.
The possibility of having alternative solutions to the first-order invariance equation hinges upon having
$c_{B^\dag} \neq 0$, by means of which a gauge
degree of freedom in the selection of the trace of terms $\bm{F}_1^{(k)} \bar{\rho}_A^{(k)}$ and $\bm{F}_2^{(k)}\bar{\rho}_A^{(k)}$
can be tuned so as to cancel out $\bm{H}_{s,1}$, as in Theorem \ref{th:manyFastFirstOrder}.
It appears that gauge choices are instrumental for positivity-preservation in the solution of the second-order invariance equation, as we consider next. }
\end{remark}
\vspace{2mm}


\begin{theorem}
\label{th:manyFastSecondOrder}
Consider model (\ref{eq:systemmanyfast}). 
Let Assumption \ref{assumption:cbassumption} hold.
Assume furthermore that $\mathcal{L}_{s,1}=\mathcal{L}_B=0$, with $\mathcal{K}_1$ selected according to
Theorem \ref{th:manyFastFirstOrder}.
Then, the second-order invariance equation is satisfied by
a Lindbladian:
{\small
\begin{align}
 \mathcal{L}_{s,2}(\rho_s) =  \sum_{k} & -i \,  \Im \left( z_1^{(k)} \right) \left[ \bm{B} \bm{B}^\dag , \rho_s \right]  
			         -i \,  \Im\left(z_2^{(k)}\right) \left[ \bm{B}^\dag \bm{B} , \rho_s \right]  \nonumber \\
& 				 +2 \,  \Re\left(z_1^{(k)}\right) \mathcal{D}_{\bm{B}^\dag}(\rho_s)
				 +2 \,  \Re\left(z_2^{(k)}\right) \mathcal{D}_{\bm{B}}(\rho_s)  \nonumber \\
& + \sum_{k>k^\prime} \Big\{ -i \; \delta^{(k,k^\prime)} \;
                                         \left[ \left[ \bm{B}, \bm{B}^\dag \right]  , \rho_s \right] 
                                         \Big\} , \label{eq:Ls2many}
\end{align}}

{\parskip = -2mm
\noindent
with:}
{\small \begin{subequations}
\begin{align}
\delta^{(k,k^\prime)} & = \frac{-2 \Re  \left( z_0^{(k)} z_0^{(k^\prime)\, \ast } \right) }{c_{B^\dag}} \label{eq:definitionsdelta} \; , \quad
z_0^{(k)} = \trace\left( \bm{A}^{(k)} \bar{\rho}_A^{(k)} \right) , \\
z_1^{(k)} & = \trace\left( \bm{F}_1^{(k)} \bar{\rho}_A^{(k)} \bm{A}^{(k)\,\dag} \right) ,\quad
z_2^{(k)} = \trace \left( \bm{F}_2^{(k)} \bar{\rho}_A^{(k)} \bm{A}^{(k)} \right) \label{eq:definitionsz2} ,
\end{align} \label{eq:definitionsz}
\end{subequations}}

{\parskip = -2mm
\noindent
and by a map $\mathcal{K}_2$, obtained from formulas (\ref{eq:kmap2many})-(\ref{eq:Uequationall}), such that $\mathcal{K}(\rho_s) = \mathcal{K}_0(\rho_s) + \varepsilon \mathcal{K}_1(\rho_s)
+ \varepsilon^2 \mathcal{K}_2(\rho_s)$ is a CPTP map up to third-order terms.}
\end{theorem}
\begin{proof}
see Appendix \ref{ssec:manyFastSecondOrder}.
\end{proof}


 \section{Application}
\label{sec:application}

\begin{figure*}[htp]
\begin{center}

\hspace{0mm}
\includegraphics[width=0.48\textwidth,trim=0 0mm 0mm 0mm]{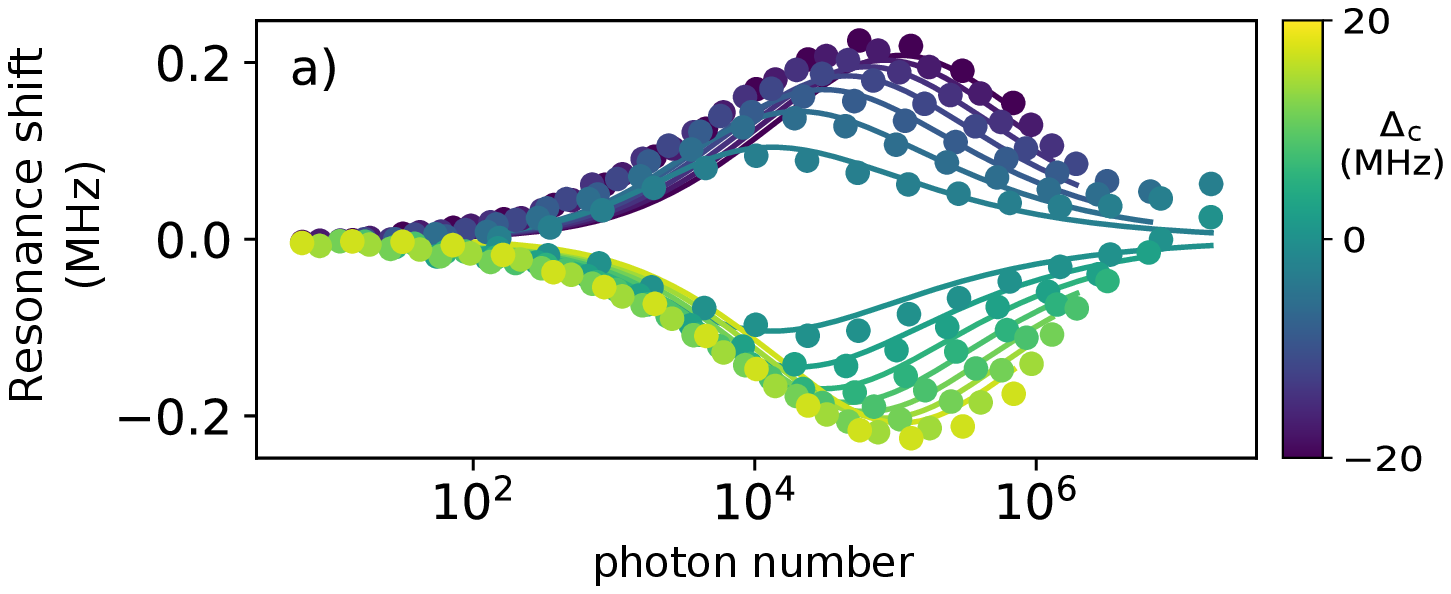}
\includegraphics[width=0.436\textwidth,trim=0 0 0mm 0]{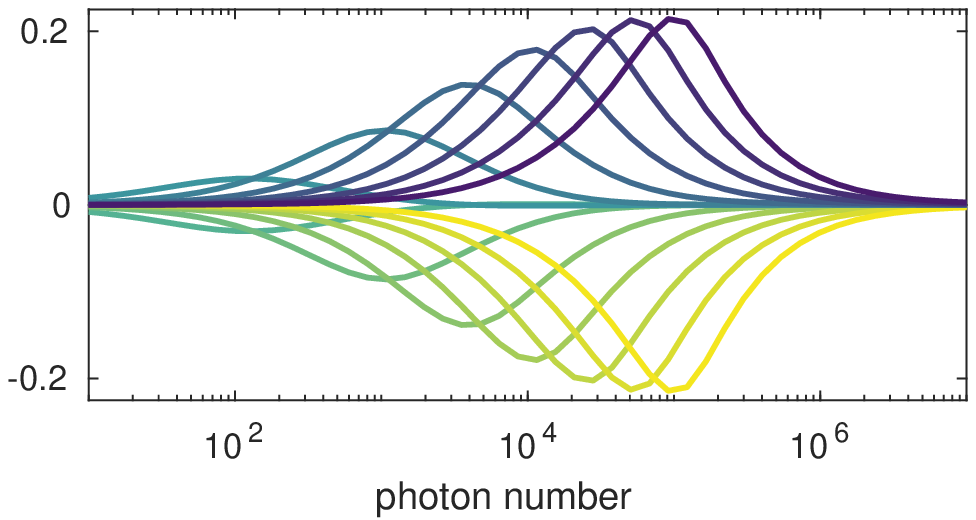}    
\caption{Section \ref{sec:application}. 
Shift of the resonator frequency: probe experiment (above left) versus reduced-order model
(\ref{eq:applicationreducedordermodel})
(above right) as a function of intraresonator photons $\left< \bm{N} \right>$ and for different pump detunings $\Delta_c$.
Parameter values in reduced-order model
(\ref{eq:applicationreducedordermodel}):
$g = \SI{30}{\kilo\hertz}$, $\Gamma_- = \SI{10}{\mega\hertz}$, $\Delta_c$ ranges from $-\SI{20}{\mega\hertz}$
to $\SI{20}{\mega\hertz}$, $\Delta_q^{(k)}$ is uniformly distribuited over $k$ in the range
$[-\SI{100}{\mega\hertz},\SI{100}{\mega\hertz}]$, and 
$\left< \bm{N} \right> = \tilde{v}^2/\Delta_c^2$ where $\tilde{v}$ is in the range from $0$ to $\SI{10}{\giga\hertz}$.}
\label{fig:application}
\end{center}
\end{figure*}

Microwave superconducting resonators are an important component in various quantum devices, and in particular
in the quantum electrodynamics circuits \cite{blais2004,rau2004} that are one of the most promising current technologies towards building a quantum computer \cite{industrials}. Losses due to imperfections in amorphous materials constitute a dominant loss channel of such resonators \cite{wang2009,gao2008}, and can be represented by a bath of two-level systems (TLSs). In many practical cases, strong microwave tones are applied with significant frequency detuning with respect to the resonance frequency \cite{regalnature} in order to activate a parametric interaction between the resonator mode and another circuit degree of freedom. Within this framework, the LKB team has performed a pump probe experiment \cite{capellearxiv} on a microwave resonator:
a strong ``pump'' drive, at a frequency far detuned from the resonator, is applied to essentially scramble the quantum behavior (``saturate'') of the TLS bath, whereas a weak probe tone, assumed not to disturb the bath behavior, is used to retrieve the transmission spectrum of the resonator. The latter allows to extract induced detuning and damping rate.


Let $\mathcal{H}_Q^{(k)}$ and $\mathcal{H}_C$ respectively be the Hilbert space of the $k$-th TLS=qubit and the resonator, and $\mathcal{H}_Q := \bigotimes_{k}{\mathcal{H}_Q^{(k)}}$.
Respectively denote with $\bm{\sigma}_+^{(k)}$ and $\bm{\sigma}_-^{(k)}$ the raising and lowering operator on
the $k$-th qubit, and with $\bm{\sigma}_x^{(k)}$, $\bm{\sigma}_y^{(k)}$, and $\bm{\sigma}_z^{(k)}$ 
the Pauli operators on the $k$-th qubit. Let $\bm{a}$ and $\bm{a}^\dag$ be the annihilation and creation
operators in the resonator mode. The experimental setup is modeled by the following system in Lindblad form:
{\small \begin{eqnarray*}
 \tfrac{d}{dt}\tilde{\rho} & = & -i [\mathbf{H},\tilde{\rho}] + \Gamma_- \, \sum_{k}{ \mathcal{D}_{\bm{\sigma}_{-}^{(k)}}(\tilde{\rho}) }
 \label{eq:sys} ,\\
\mathbf{H} &=& \omega_c \, \mathbf{a}^\dag\mathbf{a} +
		\left( v e^{i \omega_r t} + v^\ast e^{-i \omega_r t} \right) \left( \bm{a}^\dag + \bm{a} \right) \\ && +
		\sum_{k} \Big( \frac{\omega_q^{(k)}}{2} \bm{\sigma}_z^{(k)}
                  + i \, g  \, \bm{\sigma}_x^{(k)} \left(  \mathbf{a}^\dag - \mathbf{a} \right)   \Big)  \; .
\end{eqnarray*}}
Here $\omega_c$, $\omega_r$, and $\omega_q^{(k)}$ are the frequencies of the resonator, the pump drive, and the $k$-th qubit respectively, $v$ is the amplitude of the pump, $g$ is the coupling strength between the resonator 
and each qubit, and $\Gamma_-$ is the dissipation rate associated to $\bm{\sigma}_-^{(k)}$ on each qubit.
The goal would be to obtain a reduced order model for (\ref{eq:sys}) which matches the transmission spectrum of this experiment,
as in Figure \ref{fig:application}.

For each $k$, let 
$\Delta_q^{(k)} = \omega_q^{(k)} - \omega_r$ and $\Delta_c = \omega_c - \omega_r$.
Under the assumption that
$ \omega_q^{(k)} ,\,  \omega_c,\,  \omega_r \; \gg \; | \Delta_q^{(k^\prime)}|, \, | \Delta_c|,  \, g, \, \Gamma_- , \, \Gamma_+$
for any $k,k^\prime$,
we apply the standard rotating-wave approximation (i.e. first-order averaging) with
$\bm{H}_0 := \omega_r \, \mathbf{a}^\dag\mathbf{a} + \sum_{k}{\frac{\omega_r}{2} \bm{\sigma}_z^{(k)}} $ the Hamiltonian corresponding to the rotating change of frame and
$\bm{H}_1 = \bm{H}-\bm{H}_0$
the remaining Hamiltonian. The first-order RWA yields:
\begin{equation}
 \tfrac{d}{dt}{\rho}_1^{rwa} = -i \left[ \bm{H}_1^{rwa} , \, \rho_1^{rwa} \right]  + \Gamma_- \, \sum_{k}{\mathcal{D}_{\bm{\sigma}_-^{(k)}}(\rho_1^{rwa})}  , \nonumber 
\end{equation}
where 
$
 \bm{H}_1^{rwa} = \Delta_c \, \mathbf{a}^\dag\mathbf{a}
 + (v \bm{a} + v^\ast \bm{a}^\dag) + \sum_{k}{\Big( \frac{\Delta_q^{(k)}}{2} \bm{\sigma}_z^{(k)}
 + i g \left( \bm{\sigma}_-^{(k)} \bm{a}^\dag - \bm{\sigma}_+^{(k)} \bm{a} \right) \Big) }$
is the Jaynes-Cummings interaction Hamiltonian plus drive. We next apply a unitary coordinate change
$ \rho = \tilde{\bm{U}}\, \rho_1^{rwa} \, {\tilde{\bm{U}}}^\dag$ on the resonator state, to center it around its well-known steady state under off-resonant drive, namely by a complex field amplitude displacement $\tilde{\bm{U}}:= \exp \left(  \left( v^\ast \bm{a}^\dag - v \bm{a} \right)/\Delta_c \right)$. This yields:
\begin{equation}
 \tfrac{d}{dt} \rho =  \sum_{k}{ \Big\{ \mathcal{L}_Q^{(k)}(\rho) + g \mathcal{L}_{int}^{(k)}(\rho) \Big\} }    + (-i) \left[ \Delta_c \, \bm{a}^\dag \bm{a} , \; \rho \right]   , \label{eq:fullsys} 
\end{equation}
where:
\begin{align}
 \mathcal{L}_Q^{(k)}(\rho) & := -i \left[ \frac{\Delta_q^{(k)}}{2} \bm{\sigma}_z^{(k)} + \frac{g \,  \tilde{v}}{\Delta_c} \bm{\sigma}_x^{(k)}, \; \rho \right] + \Gamma_- \mathcal{D}_{\bm{\sigma}_-^{(k)}}(\rho) \nonumber , \\
\mathcal{L}_{int}^{(k)}(\rho) & := \left[  \bm{\sigma}_-^{(k)} \bm{a}^\dag - \bm{\sigma}_+^{(k)} \bm{a} , \; \rho \right] . \nonumber 
\end{align}
The term with $v := i \tilde{v}$ for $\tilde{v} \in \mathbb{R}$ now expresses an effective, indirect drive on the TLSs.

System (\ref{eq:fullsys}) is in the form (\ref{eq:systemmanyfast})
with $\mathcal{L}_A^{(k)} := \mathcal{L}_Q^{(k)}$,  $\;\;\bm{A}^{(k)} := i \bm{\sigma}_-^{(k)}$, $\bm{B} := \bm{a}$, $\;\;\tilde{\bm{H}}_{B} := \Delta_c \bm{a}^\dag \bm{a}$. The hypothesis of Theorems \ref{th:manyFastFirstOrder}
and \ref{th:manyFastSecondOrder} are satisfied since $\mathcal{L}_B = 0$ and Assumption \ref{assumption:cbassumption} holds with $c_{B^\dag} = \Delta_c$.
Despite the fact that the microwave resonator in consideration is an infinite-dimensional system and our theory is fully developed for finite-dimensional
ones, we still hope to get useful insights by applying our theory to this setup.
We will later explain how infinite dimensionality affects our results.
For now, let $\mathcal{H}_s$ be a Hilbert space whose dimension matches the dimension of the resonator space $\mathcal{H}_C$, and $\rho_s$ the density operator on $\mathcal{H}_s$. By virtue of Theorem \ref{th:manyFastSecondOrder}, the reduced model is given in Lindbladian form as follows:
 {\small
 \begin{align}
  \tfrac{d}{dt}{\rho}_s = & -i \, \left( \Delta_c + g^2 \sum_{k}{\delta^{(k)}} \right) \, \left[ \bm{a}^\dag \bm{a} , \, \rho_s\right] \label{eq:applicationreducedordermodel} \\
  & \; + 
   g^2  \left( \sum_{k}{   \Gamma_{\bm{a}}^{(k)} } \right) 
    \, \mathcal{D}_{\bm{a}}(\rho_s) + 
    g^2 \left( \sum_{k}{\Gamma_{\bm{a}^\dag}^{(k)}}\right) \, \mathcal{D}_{\bm{a}^\dag}(\rho_s)  ,  \nonumber\\
\delta^{(k)} &=   \Im \left( z_1^{(k)} + z_2^{(k)}  \right),  \nonumber\\
\Gamma_{\bm{a}^\dag}^{(k)} &=  2 \Re \left( z_1^{(k)} \right) \; , \quad  \Gamma_{\bm{a}}^{(k)} = 2 \Re \left( z_2^{(k)} \right) \; , \nonumber\\
z_1^{(k)} &=  \trace \left( -i \bm{F}_1^{(k)} \bar{\rho}_Q^{(k)} \bm{\sigma}_{+}^{(k)} \right) \; , \quad z_2^{(k)} =  \trace \left( i \bm{F}_2^{(k)} \bar{\rho}_Q^{(k)} \bm{\sigma}_{-}^{(k)} \right) \nonumber
\end{align}} and
where, for each $k$, matrices $\bm{F}_1^{(k)}$, $\bm{F}_2^{(k)}$ satisfy
equations (\ref{eq:singleFastFirstOrderSelectionOneMany}).
The solution of such equations can be computed directly since, on a qubit space
$\mathcal{H}_Q^{(k)}$, one can always parameterize operators in terms of Pauli matrices.
We immediately find:
$z_1^{(k)} = \frac{W_1^{(k)}}{Z^{(k)}}$ and $z_2^{(k)} = \frac{W_2^{(k)}}{Z^{(k)}}$ where
{\small
 \begin{align}
 W_1^{(k)} =& -4 g^2 v^2 \big( 8 i g^2 v^2 + \nonumber \\
 & (\Gamma_- + i \Delta_c) \Delta_c
                   \left( (\Gamma_- + 2 i \Delta_c )^2 +4 \Delta_q^2  \right) \big) \nonumber , \\
 W_2^{(k)} =& 32 i g^4 v^4 + 2 i (\Gamma_- - i \Delta_c) \Delta_c^4 (\Gamma_2+4\Delta_q^2) \cdot \nonumber \\
  & \cdot (\Gamma_- - 2 i (\Delta_c + \Delta_q)) - 4 g^2v^2 \Delta_c \big( 
  \Gamma_-^3 - 5 i \Gamma_-^2 \Delta_c  + \nonumber \\
  & 4 i \Delta_c (\Delta_c^2 + 2 \Delta_c \Delta_q - \Delta_q^2)
  + 4 \Gamma_- (-\Delta_c^2+ \Delta_q^2)  
  \big),
 \nonumber \\
  Z^{(k)} =&  \left( 8 g^2v^2+\Delta_c^2 (\Gamma_-^2+4 \Delta_q^2 )  \right) 
  \big( 8 g^2 v^2 ( i \Gamma_- + 2 \Delta_c ) + \nonumber \\
     & \Delta_c^2 (i \Gamma_- + \Delta_c )
      \left(  (\Gamma_- - 2 i \Delta_c)^2 + 4 \Delta_q^2     \right)    
      \big) .
 \nonumber 
 \end{align}}
Coefficients $g^2 \sum_{k}{\delta^{(k)}}$, $g^2 \sum_{k}{   \Gamma_{\bm{a}}^{(k)} }$, and
$g^2  \sum_{k}{\Gamma_{\bm{a}^\dag}^{(k)}}$ appearing in our reduced-order model
(\ref{eq:applicationreducedordermodel}) can be visualized for different values of pump detuning
$\Delta_c$ and intraresonator photon number $\left< \bm{N} \right> = \tilde{v}^2/(\kappa^2+\Delta_c^2)$.
As depicted in Figure \ref{fig:application}, we can compare the frequency shift of the resonator
$g^2 \sum_{k}{\delta^{(k)}}$ with
experimental findings from the pump probe experiment. We find that,
by properly calibrating the values of $g$ and by selecting a proper density function of the 
TLSs, we are able to match the resonance shift in the trasmission spectrum observed in the pump probe experiment.
However, we have observed that for large drive gains $\tilde{v}$ the second-order dissipation gives
$g^2 \sum_{k}{   \Gamma_{\bm{a}^\dag}^{(k)}} > g^2 \sum_{k}{   \Gamma_{\bm{a}}^{(k)}}$
which would imply that we constantly keep adding energy and the resonator state drifts off to infinity.
Our finite-dimensional treatment can obviously not be trusted in this case.
Anyway, quantitative agreement between data and our model should hence enable to
extract characteristics about the TLS bath, pending other experimental features that will have to be calibrated.

\section{Conclusions}

We have studied adiabiatic elimination for open quantum systems in Lindblad form composed by a target
subsystem
weakly interacting with $K$ strongly dissipative subsystems. The key novel features of our approach
are: (i) the decomposition of the environment into $K$ separately treated subsystems; (ii) the
presence of fast Hamiltonian dynamics on the target system. The time-scale separation between the 
uncoupled dynamics and the interaction allows model reduction via 
geometric singular perturbation theory. We have provided formulas for the first- and second-order 
expansion and shown that the asymptotic expansion of the center manifold retains a physical 
interpretation: (i) the reduced model evolves according to Lindbladian dynamics; (ii) reduced and 
original model are related via Kraus map. Each strongly dissipative subsystem contributes 
linearly to the reduced model at first-order, and does the same at second-order if a specific 
commutation property about the interaction terms holds. 
We have applied our proposed theory to the model of
a microwave superconducting resonator subject to dielectric losses
where the shape of the trasmission spectrum of our reduced-order model 
matches experimental data.
Future work will address: (i) the necessary conditions to satisfy the invariance equation at orders 
higher than two; (ii) a thorough study of the infinite-dimensional case;
(iii) a full generalization of the proposed theory by removing the assumption about the commutator
between the original Hamiltonian dynamics of the target and the interaction terms.

\appendix

\section{Proofs of Theorems}
\label{sec:proofs}

\subsection{Proof of Theorem \ref{th:manyFastFirstOrder}}
\label{ssec:manyFastFirstOrder}
Let $\bar{\rho}_A^{[k]}$ denote $\bigotimes_{k^\prime \neq k}{\bar{\rho}_A^{(k)}} $.
By plugging (\ref{eq:zeroorderK}) and (\ref{eq:K1shape00}) into the first-order invariance condition (\ref{eq:firstorderinvariance})
 and by making use of Assumption \ref{assumption:cbassumption}, condition (\ref{eq:firstorderinvariance}) reads as:
 {\footnotesize
 \begin{align}
& \sum_{k} \bar{\rho}_A^{[k]} \otimes  \Big\{ \nonumber \\
& \Big( -i \mathcal{L}_A^{(k)}\left( \bm{F}_1^{(k)} \bar{\rho}_A^{(k)}  \right)  -i \bm{A}^{(k)} \bar{\rho}_A^{(k)} - \, c_{B^\dag}\, \bm{F}_1^{(k)} \bar{\rho}_A^{(k)} \Big) \otimes \bm{B}^\dag \rho_s  \nonumber \\
&   + \Big( -i \mathcal{L}_A^{(k)}\left( \bm{F}_2^{(k)} \bar{\rho}_A^{(k)}  \right)  -i \bm{A}^{(k)\,\dag} \bar{\rho}_A^{(k)} 
  +  \, c_{B^\dag} \, \bm{F}_2^{(k)} \bar{\rho}_A^{(k)} \Big)  \otimes \bm{B} \rho_s  \nonumber \\
& + \text{\emph{herm.conj.}} \Big\}  \;+\; \Big( \mathcal{L}_B(\rho_s) -  \mathcal{L}_{s,1}(\rho_s)\Big) \otimes \bar{\rho}_A = 0 
 \end{align} \label{eq:firstorderexpandedmany}}
 It can be proved along the lines \cite[Lemma 4]{azouitQST2017} that
equations (\ref{eq:firstorderexpandedmany})
together with trace condition
(\ref{eq:conditiontrace1many}) (or (\ref{eq:conditiontrace2many}))
are always solvable for $\bm{F}_1^{(k)}, \bm{F}_2^{(k)}$.

\emph{Case of Theorem \ref{th:manyFastFirstOrder}}. 
It can be immediately seen
from (\ref{eq:firstorderexpandedmany}) that one can select
$\mathcal{L}_{s,1} := \mathcal{L}_B$ as long as the two round parenthesis 
in (\ref{eq:firstorderexpandedmany}) are set to zero.
Taking the trace on (\ref{eq:singleFastFirstOrderSelectionOneMany1}) and (\ref{eq:singleFastFirstOrderSelectionOneMany2}) yields:
{\small \begin{equation}
 c_{B^\dag} \trace \left( \bm{F}_1^{(k)} \bar{\rho}_A^{(k)} \right) = c_{B^\dag} \trace \left( \bm{F}_2^{(k)} \bar{\rho}_A^{(k)} \right)^\ast = - i \trace \left( \bm{A}^{(k)} \bar{\rho}_A^{(k)} \right) \label{eq:conditiontrace1many},
\end{equation}
which solves the situation with the announced formulas.}

\emph{Case of Remark \ref{remark:firstorder}}. 
By taking the trace on equations (\ref{eq:singleFastFirstOrderSelectionTwoMany}), we observe that, for each $k$:
\begin{equation}
 \trace \left( \bm{F}_1^{(k)} \bar{\rho}_A^{(k)} \right) = \trace \left( \bm{F}_2^{(k)} \bar{\rho}_A^{(k)} \right)^\ast = 0 \label{eq:conditiontrace2many}.
\end{equation}
Then, by taking the partial trace over $\mathcal{H}_A$ in (\ref{eq:firstorderexpandedmany}), we immediately have
$\mathcal{L}_{s,1}(\rho_s) = \mathcal{L}_B(\rho_s) -i [ \bm{H}_{s,1}, \rho_s]$ with $\bm{H}_{s,1}$ as in Theorem \ref{th:manyFastFirstOrder}.
Now, plugging $\mathcal{L}_{s,1}$ in 
(\ref{eq:firstorderexpandedmany}) yields:
 {\footnotesize
 \begin{align}
& - \sum_{k} \bar{\rho}_A^{[k]} \otimes \nonumber \\
& \Big\{   \Big( i \mathcal{L}_A^{(k)} \left( \bm{F}_1^{(k)} \bar{\rho}_A^{(k)}  \right) 
	  + i \mathcal{S}^{(k)}\left( \bm{A}^{(k)} \bar{\rho}_A^{(k)} \right)
	  + \, c_{B^\dag}\, \bm{F}_1^{(k)} \bar{\rho}_A^{(k)} \Big) \otimes \bm{B}^\dag \rho_s  \nonumber \\
&   + \Big( i \mathcal{L}_A^{(k)}\left( \bm{F}_2^{(k)} \bar{\rho}_A^{(k)}  \right) 
    +i \mathcal{S}^{(k)} \left( \bm{A}^{(k)\,\dag} \bar{\rho}_A^{(k)} \right)
    - c_{B^\dag}^\ast\, \bm{F}_2^{(k)} \bar{\rho}_A^{(k)} \Big) \otimes \bm{B} \rho_s  \nonumber  \\
 & + \text{\emph{herm.conj.}} \Big\} = 0  \label{eq:firstorderexpandedbmany}
 \end{align}}
In order to solve (\ref{eq:firstorderexpandedbmany}) it is enough to zero the two round parenthesis of (\ref{eq:firstorderexpandedbmany})
for each $k$ --- see equations (\ref{eq:singleFastFirstOrderSelectionTwoMany}).

It is immediate to see from \eqref{eq:kmap10} that $\mathcal{K}_0(\rho_s) + \varepsilon \mathcal{K}_1(\rho_s)$ is a completely positive map, as long as one can neglect the terms of order $\varepsilon^2$. One concludes that it is also trace-preserving at order $\epsilon$ by checking that $\;\trace \left( \mathcal{K}_1(\rho_s)\right) = 0\;$, thanks to (\ref{eq:conditiontrace1many})
for the case of Theorem \ref{th:manyFastFirstOrder}
and thanks to (\ref{eq:conditiontrace2many}) for the case of Remark \ref{remark:firstorder}.

\subsection{Proof of Theorem \ref{th:manyFastSecondOrder}}
\label{ssec:manyFastSecondOrder}

We prove Theorem \ref{th:manyFastSecondOrder}
in five steps:
\begin{enumerate}
 \item we define the mapping $\mathcal{K}$ up to third-order terms and, by collecting powers of $\varepsilon$,
 we obtain the formulations of $\mathcal{K}_0, \mathcal{K}_1, \mathcal{K}_2$;
 \item we decompose $\mathcal{K}_2$ in three terms and show that they satisfy the second-order
 invariance condition (\ref{eq:secondorderinvariance});
 \item by taking the partial trace w.r.t. $\mathcal{H}_A$ on (\ref{eq:secondorderinvariance}),
 we compute $\mathcal{L}_{s,2}$;
 \item we show that mapping $\mathcal{K}$, according to our definition, is a completely positive mapping;
 \item finally, we show trace-preservation of $\mathcal{K}$ by proving $\trace[\mathcal{K}_2(\rhos)]=0$
 for all $\rhos$. 
\end{enumerate}

Let $\bar{\rho}_A^{[k,k^\prime]}$ denote $\bigotimes_{k^{\prime \prime} \neq k,k^\prime}{\bar{\rho}_A^{(k^{\prime \prime})}} $.
Let $\bar{\rho}_A^{(k,k^\prime)}$ denote $\rhoa^{(k)} \otimes \rhoa^{(k^\prime)}$.
We use the following notation for operators: if the superscript of an operator respectively is
${(k)}$ or ${(k,k^\prime)}$, then it respectively applies to $\mathcal{H}_A^{(k)}$ or $\mathcal{H}_A^{(k)} \otimes \mathcal{H}_A^{(k^\prime)}$ only, possibly leaving the remaining subsystems identical;
superscripts $(k,k)$ or $[k,k]$ are used interchangeably with superscripts $(k)$ and $[k]$ respectively.
Let $\mathcal{S}^{(k,k^\prime)}(Q)$ denote $Q-\trace(Q)\bar{\rho}_A^{(k,k^\prime)}$.
Let $\bm{B}_1 := \bm{B}^\dag, \bm{B}_2 := \bm{B}$, $\bm{A}_1^{(k)} := \bm{A}^{(k)}, \bm{A}_2^{(k)} := \bm{A}^{(k)\, \dag}$.
For any $k,k^\prime$, let $\{ \bm{U}_{B_j^\dag B_h}^{(k,k^\prime)}\}_{h,j \in \{1, 2\}}$ be a family of four operators 
on $\mathcal{H}_A^{(k)} \otimes \mathcal{H}_A^{(k^\prime)}$ only, which we will define in the following.
Define:
{\small
\begin{align}
\bm{M} = & \sum_{k} \; \sum_{j \in \{ 1, 2\}}{ \bm{F}_j^{(k)} \otimes \bm{B}_j} \nonumber \\
\bm{N} := & \sum_{k,k^\prime} \; \sum_{h,j \in \{1, 2\}} \bm{U}_{B_j^\dag B_h}^{(k,k^\prime)} \otimes \bm{B}_j^\dag \bm{B}_h
\nonumber \\
\bm{W}_\mu^{(k)}(t) := &    [\bm{L}_{A,\mu}^{(k)}, \, \bm{F}_1^{(k)}] \otimes \bm{B}^\dag +
			      e^{2 i c_{B^\dag} t}[\bm{L}_{A,\mu}^{(k)}, \, \bm{F}_2^{(k)}] \otimes \bm{B}   \nonumber \\
\mathfrak{w}^{(k)}(t) := & \bm{B}_1 + b^{(k)} e^{-2ic_{B^\dag} t } \bm{B}_2 \nonumber \\
b^{(k)} := & - \frac{1}{2 c_{B^\dag }} \, \trace \left[ \sum_{\mu} [ \bm{L}_{A,\mu}, \, \bm{F}_2^{(k)} ] \, \rhoak  \, [ \bm{L}_{A,\mu}, \, \bm{F}_1^{(k)}]^\dag \right] 
\label{eq:kmap2many} .
\end{align}}
Let $b_1^{(k)} := b^{(k) \, \star}$ and $b_2^{(k)} := b^{(k)}$,
Let $\mathcal{L}_A^{(k,k^\prime)}$ denote the operator $\mathcal{L}_A^{(k)}+\mathcal{L}_A^{(k^\prime)}$.
Let $\delta_{hj} := h-j$. Let $f_1, f_2 \ge 0$ two constants which we will define in the following. Now define:
{\small \begin{align}
\mathfrak{g}(\rhos) := &  \sum_{k}{ \int_{0}^{\pi / (2 c_{B^\dag})} \mathfrak{w}^{(k)}(t) \, \rhos \, \mathfrak{w}^{(k)}(t)^\dag \,\, dt    }  \nonumber \\
\mathfrak{f}(\rhos) := & \sum_{j \in \{1, 2\}}{f_j \bm{B}_j \rhos \bm{B}^\dag_j } \nonumber \\
\mathcal{G}(\rhos) := &  \trace_A \left[ \int_{0}^{\pi / c_{B^\dag}}{ \sum_{k, \mu}  \, \bm{W}_\mu^{(k)}(t) (\rhoa \otimes \rhos) \bm{W}_\mu^{(k)}(t)^\dag \,\, dt } \right] \nonumber \\
   \mathcal{K}_2^Q( \rhos ) := & \int_{0}^{+\infty}{ e^{\mathcal{L}_A(\cdot)t} \Big( \mathcal{S} \big( \sum_{k, \mu}  \,  \bm{W}_\mu^{(k)}(t) (\rhoa\otimes \rhos) \bm{W}_\mu^{(k)}(t)^\dag \big) \Big) \,\, dt     }
          \nonumber \\
 & + \frac{c_{B^\dag} \bar{\tau} }{\pi}   \rhoa \otimes \mathcal{G}(\rhos) + c_{B^\dag} \, \rhoa \otimes \mathfrak{g}(\rhos) + \rhoa \otimes \mathfrak{f}(\rhos)   \nonumber \\
 \mathcal{K}(\rho_s) := &  \left( \bm{I} -i \varepsilon \bm{M} + \varepsilon^2 \bm{N} \right) \left(\bar{\rho}_A \otimes \rho_s \right)
 \left( \bm{I} +i \varepsilon \bm{M}^\dag + \varepsilon^2 \bm{N}^\dag \right)  \nonumber  \\
&  + \varepsilon^2 \mathcal{K}_2^Q( \rhos )   \label{eq:kmap1many} .
\end{align}}
By collecting powers of $\varepsilon$ in (\ref{eq:kmap1many}) and carrying out straightforward computations,
we obtain the formulation $\mathcal{K}(\rho_s) = (\mathcal{K}_0 + \varepsilon \mathcal{K}_1 + \varepsilon^2 \mathcal{K}_2)(\rho_s)$ where:
{\small \begin{align}
 \mathcal{K}_0(\rho_s)  = &  \bar{\rho}_A \otimes \rho_s, \label{eq:K0shapemany} , \\
 \mathcal{K}_1(\rho_s)  = & -i \; \sum_{k^\prime} \,\,  \sum_{j \in \{1, 2\}} \; \bar{\rho}_A^{[k^\prime]} \otimes  \bm{F}_j^{(k^\prime)}   \bar{\rho}_A^{(k^\prime)}   \otimes \bm{B}_j \rhos   \, + \, \hc \label{eq:K1shapemany} , \\
\mathcal{K}_2(\rhos) = & \mathcal{K}_2^L(\rhos) + \mathcal{K}_2^E(\rhos)  + \mathcal{K}_2^Q(\rhos)    \label{eq:K2shapemany} ,
\end{align}}
where: 
{\small \begin{align}
\mathcal{K}_2^L(\rhos )  = & \sum_{k,  k^\prime} \, \sum_{h,j \in \{ 1, 2 \}} \Big\{  \rhoalesskk \otimes \bm{U}_{B_j^\dag B_h}^{(k,k^\prime)} \rhoakk \otimes \bm{B}_j^\dag \bm{B}_h \rhos  \Big\} + \hc , \nonumber \\
 \mathcal{K}_2^E(\rhos ) = & \sum_{k,  k^\prime} \, \sum_{h,j \in \{ 1, 2 \}} \Big\{ \rhoalesskk \otimes \bm{F}_h^{(k)} \rhoakk \bm{F}_j^{(k^\prime)\,^\dag} \otimes \bm{B}_h \rhos \bm{B}_j^\dag \Big\}  \nonumber \\
\mathcal{K}_2^Q(\rho_s)  = & \sum_{k} \Big\{ \sum_{h, j \in \{1, 2\}} \; \Big(   \rhoalesskk \otimes \bar{\mathcal{K}}_{hj}^{(k)}(\rhos) \otimes \bm{B}_h \rhos \bm{B}_j^\dag \Big) \nonumber \\
 & \;\;\;\;\;\; + \sum_{j \in \{1, 2\}} \Big(  \bar{\tau} \,  \trace \left[ \bar{\mathcal{F}}_{jj}(\rhoa)^{(k)}  \right] \rhoa  \otimes \bm{B}_j \rhos \bm{B}_j^\dag    \Big)  \Big\} \nonumber \\
 & + \sum_{k} \,  \rhoa \otimes  \Big\{ \frac{\pi }{2}  \left( \bm{B}^\dag \rhos \bm{B} + |b_2^{(k)}|^2 \bm{B} \rhos \bm{B}^\dag \right) \nonumber \\
 & \;\;\;\;\;\;\;\;\;\;\;  - \sum_{h \neq j} i  \delta_{hj} b_h^{(k)} \bm{B}_h \rhos \bm{B}_j^\dag \Big\} + \nonumber \\
 & + \rhoa \otimes  \left(  f_1 \bm{B}^\dag \rhos \bm{B} + f_2 \bm{B} \rhos \bm{B}^\dag  \right)  ,
			        \label{eq:Kq2expanded} 
\end{align}}
and where:
{\footnotesize
\begin{align}
 \bar{\mathcal{K}}_{h j}^{(k)}(\rhos)  & := \int_{0}^{+\infty} e^{t \, \mathcal{L}_A^{(k)}(\cdot)} \Big( \mathcal{S}^{(k)}
					  \big( \bar{\mathcal{F}}_{h j}(t,\rhoa)^{(k)}\big) \Big) \;\; dt   \nonumber \\
 \bar{\mathcal{F}}_{hj}(t,\rhoa)^{(k)} & := \exp\left( 2 i c_{B^\dag} \delta_{hj}  \, t \right) \, \sum_{\mu}
					    \left[ \bm{L}_{A,\mu}^{(k)}, \, \bm{F}_h^{(k)} \right] \,
					    \rhoak \,
					    \left[ \bm{L}_{A,\mu}^{(k^\prime)}, \, \bm{F}_j^{(k^\prime)} \right] \nonumber \\
\bar{\mathcal{F}}_{hj}(\rhoa)^{(k)} & := \bar{\mathcal{F}}_{hj}(t,\rhoa)^{(k)} \nonumber \\
f_1 & := \sum_{k,k^\prime} \trace \left[ \bm{F}_2^{(k)} \rhoakk \bm{F}_2^{(k^\prime)} + \check{\delta}_{kk^\prime} \left( \bar{\mathcal{F}}_{22}^{(k)}
+ \frac{\pi}{2} | b^{(k)}|^2\right) \right] \nonumber \\
f_2 & := \sum_{k,k^\prime} \trace \left[ \bm{F}_1^{(k)} \rhoakk \bm{F}_1^{(k^\prime)} + \check{\delta}_{kk^\prime} \left( \bar{\mathcal{F}}_{11}^{(k)}
+ \frac{\pi}{2}  \right) \right]  \label{eq:K2defmore} .
\end{align}}
where $\check{\delta}_{kk^\prime}$ denotes the Kronecker delta. For any $k,k^\prime$, let the family of operators $\{ \bm{U}_{B_j^\dag B_h}^{(k,k^\prime)}\}_{h,j \in \{1, 2\}}$
satisfy the following set of equations:
{\small \begin{subequations}
\begin{align}
& \mathcal{L}_A^{(k,k^\prime)} 
\left(  \bm{U}_{B_j^\dag B_h}^{(k,k^\prime)} \, \rhoakk\right) + \mathcal{S}^{(k,k^\prime )}
\Big( 2i c_{B^\dag} \, \delta_{hj}  \, \bm{U}_{B_j^\dag B_h}^{(k,k^\prime)} \, \rhoakk - \nonumber \\
& \;\;\;\; - \bm{A}_j^{(k)\, \dag} \bm{F}_h^{(k^\prime)} \, \rhoakk
\Big) = 0 , \;\;\; \forall h,j \in \{ 1, 2 \} \label{eq:Uequation} ,  \\
& 
\trace \left[ 2i c_{B^\dag} \, \delta_{hj}  \, \bm{U}_{B_j^\dag B_h}^{(k,k^\prime)} \, \rhoakk 
 \right]  = \nonumber \\
 & \;\;\;\;\;\; = \trace \left[ \bm{A}_j^{(k)\, \dag} \bm{F}_h^{(k^\prime)} \, \rhoakk \right] 
  , \;\;\; \forall h \neq j
   \in \{ 1, 2 \} \label{eq:Uequation2} , \\
   & 
\trace \left[ \bm{U}_{B B^\dag}^{(k,k^\prime)} \rhoakk \right]  =
    - \trace \left[ \bm{F}_1^{(k)} \rhoakk \bm{F}_1^{(k^\prime)\, \dag } + \check{\delta}_{k k^\prime} \left( \bar{\tau} \bar{\mathcal{F}}_{1 1}(\rhoa)^{(k)} + \frac{\pi}{2} \right) \right] , \label{eq:Uequation3} \\
   & 
\trace \left[ \bm{U}_{B^\dag B}^{(k,k^\prime)} \rhoakk \right]  =
    - \trace \left[ \bm{F}_2^{(k)} \rhoakk \bm{F}_2^{(k^\prime)\, \dag } + \check{\delta}_{k k^\prime} \left( \bar{\tau} \bar{\mathcal{F}}_{2 2}(\rhoa)^{(k)} + \frac{\pi}{2} |b^{(k)}|^2 \right) \right] \label{eq:Uequation4} .
\end{align} \label{eq:Uequationall}
\end{subequations}}
Equations (\ref{eq:Uequationall}) are always solvable,
as proved in \cite[Lemma 4]{azouitQST2017}. 
We will then show that our definition of $\mathcal{K}_2$ in (\ref{eq:kmap2many})-(\ref{eq:Uequationall}) indeed satisfies the 
second-order invariance condition (\ref{eq:secondorderinvariance}).
We start by observing that, thanks to assumptions  $\mathcal{L}_B(\cdot) = 0$ and $\mathcal{L}_{s,1} = 0$,
condition 
(\ref{eq:secondorderinvariance}) reads as:
\begin{align}
& \mathcal{L}_A\left( \mathcal{K}_2(\rho_s)  \right) 
 + \mathcal{L}_{int} \left( \mathcal{K}_1(\rho_s) \right) -i \Big( [ \tilde{\bm{H}}_B , \, \mathcal{K}_2(\rho_s) ] \nonumber \\
 & - \mathcal{K}_2( [\tilde{\bm{H}}_B, \, \rhos] ) \Big)   \; = \; 
 \rhoa \otimes \mathcal{L}_{s,2}(\rho_s)   \label{eq:secondorderinvariance2many} .
\end{align}
Then, in order to compute the left-hand side of (\ref{eq:secondorderinvariance2many}),
we observe that 
the computation of term $-i [ \tilde{\bm{H}}_B , \, \mathcal{K}_2(\rho_s) ]-\mathcal{K}_2 (\mathcal{L}_{s,0}(\rho_s))$
is simplified  by the following set of properties directly implied by Assumption \ref{assumption:cbassumption}:
{\small
\begin{align}
& \left[ \tilde{\bm{H}}_B, \bm{B} \right] =  -c_{B^\dag} \bm{B} , \nonumber \\
& \left[ \tilde{\bm{H}}_B, \bm{B}_h \bm{B}_j^\dag  \right] =  - 2 \delta_{hj} \, c_{B^\dag} \bm{B}_h \bm{B}_j^\dag \;\;\; \forall h,j \in \{1, 2\} \nonumber , \\
& -i \left( \left[ \tilde{\bm{H}}_B, \,\, \bm{B}_h \rhos \bm{B}_j^\dag \right]  - \bm{B}_h \left[  \tilde{\bm{H}}_B, \, \rhos \right] \bm{B}_j^\dag \right)= \nonumber \\
& \;\;\;\;\;\;\;  = 2 i c_{B^\dag} \delta_{hj} \bm{B}_h \rhos \bm{B}_j^\dag \;\;\; \forall h,j \in \{1, 2\}  \label{eq:brhobproperty} ,
\end{align}}
whereas the computation of term $\mathcal{L}_A(\mathcal{K}_2(\rhos))$ is simplified by the following Claim.
\begin{claim} \label{claim:LaK2Q}
 $\mathcal{L}_A^{(k)} \left( \bar{\mathcal{K}}_{hj}^{(k)}(\rhos) \right)
 + 2 i c_{B^\dag} \delta_{hj} \bar{\mathcal{K}}_{hj}^{(k)} (\rhos)
 + \mathcal{S}^{(k)}\left(  \bar{\mathcal{F}}_{hj}(\rhoa)^{(k)}\right)=0$.
\end{claim}
\begin{proof}
 Case $h=j$ is proved along the lines of \cite[Lemma 1 and Lemma 4]{azouitQST2017}.
 Case $(h,j)=(2,1)$ is the hermitian conjugate of case $(h,j)=(1,2)$ which we are now going to prove.
 Let $\mathcal{L}^\sharp (\cdot)$ denote the super-operator
 $\mathcal{L}_A^{(k)}(\cdot) - 2 i c_{B^\dag} \textit{Id}(\cdot)$. Then, since $\mathcal{L}_A^{(k)}$ is strongly
 dissipative on $\mathcal{H}_A^{(k)}$, we have that:
\begin{equation}
\lim_{t \rightarrow +\infty}{\exp ( t \mathcal{L}^\sharp (\cdot))(\bm{X})} = \lim_{t \rightarrow +\infty}{\exp(-i 2 c_{B^\dag}t) \exp(t\,\mathcal{L}_A^{(k)}(\cdot))(\bm{X})} = 0 , \label{eq:claim2dev1} 
\end{equation}
for any operator $\bm{X}$ such that  $\trace_{\mathcal{H}_A^{(k)}}[\bm{X}] = 0$.
 First, we formulate $\bar{\mathcal{K}}_{12}^{(k)}$ as:
 \begin{align}
  & \bar{\mathcal{K}}_{12}^{(k)}(\rhos ) = \int_{0}^{+\infty}{\mathfrak{K}_{12}(t,\rhoa)^{(k)}\,\, dt} \nonumber \\
  & \mathfrak{K}_{12}(t,\rhoa)^{(k)} := \exp \left( t \mathcal{L}^\sharp (\cdot)  \right)
    \Big( \mathcal{S}^{(k)} \left( \bar{\mathcal{F}}_{12}(\rhoa)^{(k)} \right) \Big) \nonumber 
  \end{align}
Second, we observe that:
\begin{align}
 & \mathcal{L}^\sharp \left(\mathfrak{K}_{12}(t,\rhoa)^{(k)}\right) = \frac{d}{dt} \mathfrak{K}_{12}(t,\rhoa)^{(k)} \label{eq:claim2dev2} .
\end{align}
We then conclude from (\ref{eq:claim2dev1}) and (\ref{eq:claim2dev2}) that:
\begin{align}
 \mathcal{L}^\sharp \left( \bar{\mathcal{K}}_{12}^{(k)}(\rhos) \right) = 
  \left[ \mathfrak{K}_{12}(t,\rhoa)^{(k)} \right]_{0}^{+\infty} = -\mathcal{S}^{(k)}\left( \bar{\mathcal{F}}_{12}(\rhoa)^{(k)}\right) \nonumber .
\end{align}
\end{proof}
Furthermore, by considering $\mathcal{L}_{int}(\rho) = -i \sum_{k} \sum_{h \in \{1 , 2 \}} [ \bm{A}_h^{(k)\, \dag} \otimes \bm{B}_h^\dag , \,\, \rho ]$,
we have that:
{\small 
\begin{align}
 \mathcal{L}_{int}(\mathcal{K}_1(\rhos)) =& \sum_{k,k^\prime} \, \rhoalesskk \otimes \sum_{j,h \in \{ 1, 2 \}} \bm{B}_j^\dag \bm{B}_h \rhos \otimes \left( -\bm{A}_j^{(k)\, \dag} \bm{F}_h^{(k^\prime)} \rhoakk \right) 
 + \hc \nonumber \\
 & + \bm{B}_h \rhos \bm{B}_j^\dag \otimes \left( \bm{F}_h^{(k)} \rhoakk \bm{A}_j^{(k^\prime)\, \dag } + \bm{A}_h^{(k)} \rhoakk \bm{F}_j^{(k^\prime)\, \dag } \right) , \nonumber \\
 =  & \sum_{k,k^\prime} \, \rhoalesskk \otimes \sum_{j,h \in \{ 1, 2 \}} \bm{B}_j^\dag \bm{B}_h \rhos \otimes \left( -\bm{A}_j^{(k)\, \dag} \bm{F}_h^{(k^\prime)} \rhoakk \right) 
 + \hc \nonumber \\
 & + \bm{B}_h \rhos \bm{B}_j^\dag \otimes \Big(   - 2 i c_{B^\dag} \, \delta_{hj} \, \bm{F}_h^{(k)} \rhoakk \bm{F}_j^{(k^\prime)\, \dag} \nonumber \\
 & \;\;\;\;\;\;\;\;\;\; - \mathcal{L}_A^{(k,k^\prime)} \left( \bm{F}_h^{(k)} \rhoakk \bm{F}_j^{(k^\prime)\, \dag } \right)   \nonumber \\
 & \;\;\;\;\;\;\;\;\;\; + \check{\delta}_{k,k^\prime} \, \bar{\mathcal{F}}_{hj}(\rhoa)^{(k)}  \Big) , \label{eq:LintK1dev}
\end{align}}
where, in the last equality, 
we first made use of formulas (\ref{eq:singleFastFirstOrderSelectionOneMany}) and then we used the following 
formula generalized from \cite[Lemma 6]{azouitQST2017}:
{\small 
\begin{align}
& \bm{F}_h^{(k)}\mathcal{L}_A^{(k,k^\prime)} \left( \rhoakk \bm{F}_j^{(k^\prime) \,\dag} \right) 
+ \mathcal{L}_A^{(k,k^\prime)} \left( \bm{F}_h^{(k)} \rhoakk  \right) \bm{F}_j^{(k^\prime) \, \dag } = \nonumber \\
& = \mathcal{L}_A^{(k,k^\prime)} \left( \bm{F}_h^{(k)} \rhoakk   \bm{F}_j^{(k^\prime) \, \dag }  \right) 
- \check{\delta}_{k,k^\prime} \, \bar{\mathcal{F}}_{hj}(\rhoa)^{(k)} , \;\;\; \forall h,j, \; \forall k, k^\prime \label{eq:fhfjformula} .
\end{align}}
Finally, by making use of definitions (\ref{eq:kmap1many})-(\ref{eq:Kq2expanded}), properties (\ref{eq:brhobproperty}), simplification (\ref{eq:LintK1dev}), and 
Claim \ref{claim:LaK2Q}, we can compute:
\begin{align}
 & \mathcal{L}_A(\mathcal{K}_2(\rhos)) - i \big ( [ \tilde{\bm{H}}_B, \, \mathcal{K}_2(\rhos) ] - \mathcal{K}_2([\tilde{\bm{H}}_B, \, \rhos ])  \big) + \mathcal{L}_{int}(\mathcal{K}_1(\rhos)) = \nonumber \\
 & = \sum_{k,k^\prime} \rhoalesskk \otimes  \sum_{h,j \in \{1, 2\}} \Big\{ \bm{B}_j^\dag \bm{B}_h \rhos \otimes \Big(
   \mathcal{L}_A^{(k,k^\prime)}\big( \bm{U}_{B_j^\dag B_h}^{(k,k^\prime)} \rhoakk \big) \nonumber \\
   & \;\;\;\;\;\;\; + 2 i c_{B^\dag} \, \delta_{hj}  \, \bm{U}_{B_j^\dag B_h}^{(k,k^\prime)} \rhoakk  -
   \bm{A}_j^{(k) \, \dag } \bm{F}_h^{(k^\prime )} \rhoakk  \Big) + \hc  \Big\} \nonumber \\
& + 
\sum_{k} \rhoalessk \otimes  \sum_{h,j \in \{1, 2\}} \Big\{  \bm{B}_h \rhos \bm{B}_j^\dag \otimes \Big( 2 c_{B^\dag} \, |\delta_{hj}| b_h^{(k)} \rhoak \nonumber \\
& \;\;\;\;\;\;\;  + \trace \left[  \bar{\mathcal{F}}_{hj}(\rhoa)^{(k)} \right] \, \rhoak \Big) \Big\} =: \mathcal{E}(\rhos)
    \label{eq:almostLs2} .
\end{align}
Observe that the definition of $b^{(k)}_h$ in (\ref{eq:kmap2many}) implies that 
$2 c_{B^\dag} \, |\delta_{hj}| b_h^{(k)} + \trace [  \bar{\mathcal{F}}_{hj}(\rhoa)^{(k)} ] =0$
whenever $h \neq j$.
Furthermore, in the last equality of
(\ref{eq:LintK1dev}), we observed that
$\trace [ 
       \bm{F}_h^{(k)} 
       \rhoak \bm{A}_h^{(k) \, \dag } 
      + \bm{A}_h^{(k)} \rhoak \bm{F}_h^{(k)\, \dag } ]
      = \trace [ \bar{\mathcal{F}}_{hh}(\rhoa)^{(k)}]$ for any $h \in \{ 1, 2 \}$.
The latter two observations and trace condition (\ref{eq:Uequation2}) will be instrumental
in the derivation of $\mathcal{L}_{s,2}$ in (\ref{eq:Ls2almostmany}).
Indeed, by recalling (\ref{eq:secondorderinvariance2many}),
Lindblad $\mathcal{L}_{s,2}$ can be obtained by taking the partial trace over $\mathcal{H}_A$ of
$\mathcal{E}(\rhos)$ in expression (\ref{eq:almostLs2}), as follows:
{\footnotesize
\begin{align}
\mathcal{L}_{s,2}(\rho_s) =&  \sum_{k} \Big\{ \bm{B}\bm{B}^\dag \rho_s   \;\; \trace \left(
				  - \bm{A}^{(k)\,\dag} \bm{F}_1^{(k)} \bar{\rho}_A^{(k)}
				\right)
				+ \emph{herm.conj.} + \nonumber \\ &
 \bm{B}^\dag \bm{B} \rho_s  \;\; \trace 
				\left(
				  - \bm{A}^{(k)} \bm{F}_2^{(k)} \bar{\rho}_A^{(k)}
				\right) 
				+ \emph{herm.conj.} + \nonumber \\ 
				 &
 \bm{B}^\dag \rho_s \bm{B}  \;\; \trace
				\left(
				 \bm{F}_1^{(k)} \bar{\rho}_A^{(k)} \bm{A}^{(k)\,\dag}
				+ \bm{A}^{(k)} \bar{\rho}_A^{(k)} \bm{F}_1^{(k)\,\dag}
				\right)
				+ \nonumber \\ &
 \bm{B} \rho_s \bm{B}^\dag   \;\; \trace
				\left(
				 \bm{F}_2^{(k)} \bar{\rho}_A^{(k)} \bm{A}^{(k)}
				+ \bm{A}^{(k)\,\dag} \bar{\rho}_A^{(k)} \bm{F}_2^{(k)\,\dag}
				\right) \Big\} + \nonumber \\ 
 \sum_{k\neq k^\prime} \Big\{ & \bm{B}\bm{B}^\dag \rho_s   \;\;
				   \trace \left( - \bm{A}^{(k)\,\dag} \bar{\rho}_A^{(k)} \right) \trace \left( \bm{F}_1^{(k^\prime)} \bar{\rho}_A^{(k^\prime)} \right) + \emph{h.c.}
				+ \nonumber  \\ &
 \bm{B}^\dag \bm{B} \rho_s  \;\;
				   \trace \left(- \bm{A}^{(k)} \bar{\rho}_A^{(k)} \right) \trace \left( \bm{F}_2^{(k^\prime)} \bar{\rho}_A^{(k^\prime)} \right)
				+ \emph{h.c.} . 
\label{eq:Ls2almostmany}
\end{align}}
If we now apply definitions (\ref{eq:definitionsz}) and property (\ref{eq:conditiontrace1many}),
the expression of $\mathcal{L}_{s,2}$ in (\ref{eq:Ls2almostmany}) simplifies to:
{\footnotesize
\begin{subequations}
\begin{align}
\mathcal{L}_{s,2}(\rho_s) =  \sum_{k} \Big\{ & -z_1^{(k)} \;  \bm{B}\bm{B}^\dag \rho_s  
 -z_2^{(k)} \;  \bm{B}^\dag \bm{B} \rho_s   \;\; + \hc +
				 \nonumber \\ &
+ \left( z_1^{(k)} + (z_1^{(k)\,\ast} \right) \; \bm{B}^\dag \rho_s \bm{B}  
				+ \nonumber \\ &
+ \left( z_2^{(k)} + (z_2^{(k)\,\ast} \right) \;  \bm{B} \rho_s \bm{B}^\dag   \Big\} + \nonumber \\ 
+ \sum_{k\neq k^\prime} \frac{i}{c_{B^\dag}} \Big\{ & \left( \bm{B}\bm{B}^\dag \rho_s  +  \rho_s \bm{B}^\dag \bm{B} \right)  \;\;
				   \left(  z_0^{(k)\,\ast} z_0^{(k^\prime)} \right) 
				- \text{\emph{ herm.conj.}} \nonumber \\ &
+ \left( \bm{B}^\dag \bm{B} \rho_s  + \rho_s \bm{B}\bm{B}^\dag \right) \;\;
				   \left( -  z_0^{(k)} z_0^{(k^\prime)\,\ast} \right) 
 \Big\} \nonumber \\
=  \sum_{k}  \Big\{ & -i \,  \Im \left( z_1^{(k)} \right) \left[ \bm{B} \bm{B}^\dag , \rho_s \right] 
			         -i \,  \Im\left(z_2^{(k)}\right) \left[ \bm{B}^\dag \bm{B} , \rho_s \right] + \nonumber \\
& 				 +2 \,  \Re\left(z_1^{(k)}\right) \mathcal{D}_{\bm{B}^\dag}(\rho_s)
				 +2 \,  \Re\left(z_2^{(k)}\right) \mathcal{D}_{\bm{B}}(\rho_s) \Big\} - \nonumber \\ 
- i \sum_{k > k^\prime} \Big\{ & \; \delta^{(k,k^\prime)} \; \left( \bm{B}\bm{B}^\dag \rho_s  +  \rho_s \bm{B}^\dag \bm{B} \right)  \;\;
				- \text{\emph{ herm.conj.}} \Big\} \nonumber ,
\end{align} \label{eq:Ls2almostmany2ndpassage}
\end{subequations}}

{\parskip = -2mm
\noindent
which immediately reads as (\ref{eq:Ls2many}).
Now, by first subtracting $\rhoa \otimes \mathcal{L}_{s,2}(\rhos)$ from $\mathcal{E}(\rhos)$ in expression (\ref{eq:almostLs2})
and then using (\ref{eq:Uequation}), we conclude that:
{\small \begin{align}
 & \mathcal{E}(\rhos) - \rhoa \otimes \trace \left[ \mathcal{E}(\rhos) \right] = \nonumber \\
& = \sum_{k,k^\prime} \rhoalesskk \otimes  \sum_{h,j \in \{1, 2\}} \Big\{ \bm{B}_j^\dag \bm{B}_h \rhos \otimes \Big(
   \mathcal{L}_A^{(k,k^\prime)}\big( \bm{U}_{B_j^\dag B_h}^{(k,k^\prime)} \rhoakk \big) \nonumber \\
   & \;\;\;\;\;\;\; + \mathcal{S}^{(k,k^\prime)} \big( 2 i c_{B^\dag} \, \delta_{hj}  \, \bm{U}_{B_j^\dag B_h}^{(k,k^\prime)} \rhoakk  -
   \bm{A}_j^{(k) \, \dag } \bm{F}_h^{(k^\prime )} \rhoakk  \big) \Big)  + \hc  \Big\}  = 0 \nonumber ,
 \end{align}}
which immediately shows that (\ref{eq:secondorderinvariance2many}) is satisfied as an identify.}

We are now going to prove that $\mathcal{K}(\rhos)$ in our definition (\ref{eq:kmap1many}) is indeed a CPTP mapping.
Since terms 
\begin{align}
( \bm{I} - i \varepsilon \bm{M}  + \varepsilon^2 \bm{N} )
(\rhoa \otimes \rhos) 
( \bm{I} - i \varepsilon \bm{M}  + \varepsilon^2 \bm{N} )^\dag
+\varepsilon^2  \left( c_{B^\dag} \rhoa \otimes \mathfrak{g}(\rhos) + \rhoa \otimes \mathfrak{f}(\rhos)\right) \nonumber
\end{align}
already retain the Kraus map form, what remains to prove is complete positivity of $\mathcal{K}_2^Q(\rhos)$.
\begin{claim} \label{claim:K2Qcptp}
  There exists $\bar{\tau}>0$ such that $\mathcal{K}^Q_2(\rhos)$ is a completely positive mapping.
\end{claim}
\begin{proof}
Consider an Hilbert basis $\{ \ket{n} \}_{1 \le n \le d}$ of $\mathcal{H}_A \otimes \mathcal{H}_B$. Let $\tilde{\mathcal{H}}$ be any Hilbert space of finite dimension.
Let:
\begin{align}
 \bar{\mathcal{K}}_{\bar{\tau}}(\bm{X}) := & \int_{0}^{+\infty} e^{\mathcal{L}_A(\cdot) t} \Big( \mathcal{S} \big( \bm{W}_\mu^{(k)}(t) \, \bm{X} \, \bm{W}_\mu^{(k)\,\dag}(t) \big) \Big) \, dt \nonumber \\
 & + \bar{\tau} \frac{c_{B^\dag}}{\pi} \rhoa \otimes \trace_A
 \left[ \int_{0}^{\pi / {c_{B^\dag}}} \bm{W}_\mu^{(k)}(t) \, \bm{X} \, \bm{W}_\mu^{(k)\,\dag}(t)  \, dt  \right]
\end{align}
For each $n$ and $\nu$, select any $\ket{\phi_n}, \ket{\psi_\nu} \in \tilde{\mathcal{H}}$ and
define:
\begin{equation}
 \ket{\Phi} := \sum_{n=1}^{d}{\ket{n} \otimes \ket{\phi_n}}, \;\;\;
 \ket{\Psi} := \sum_{\nu=1}^{d}{\ket{\nu} \otimes \ket{\psi_\nu}}, \;\;\;
\end{equation}
We are then going to prove that there exists $\bar{\tau}>0$ such that 
$\bra{\Psi} \bar{\mathcal{K}}(  \ket{\Phi} \bra{\Phi} ) \ket{\Psi} \ge 0$.
Standard computations give:
{\small \begin{align}
 \bra{\Psi} \bar{\mathcal{K}}(  \ket{\Phi} \bra{\Phi} ) \ket{\Psi} \ge 0 = &
 \sum_{n^\prime, \nu^\prime, n, \nu} z_{n^\prime, \nu^\prime}^\star \, M_{n^\prime, \nu^\prime, n, \nu} \, z_{n, \nu} \nonumber ,
\end{align}}
where $z_{n,\mu} := \braket{\phi_n}{\psi_\nu }$ and
{\small \begin{align}
 M_{n^\prime, \nu^\prime, n, \nu} := & \int_{0}^{+\infty}{  \mathfrak{m}_{n^\prime, \nu^\prime, n, \nu}(t) - \mathfrak{r}_{n^\prime, \nu^\prime, n, \nu}(t) \, dt } + \bar{\tau}\,  \frac{c_{B^\dag}}{\pi} R_{n^\prime, \nu^\prime, n, \nu}    \nonumber \\
 R_{n^\prime, \nu^\prime, n, \nu} := & \int_{0}^{\pi/ c_{B^\dag}}{ \mathfrak{r}_{n^\prime, \nu^\prime, n, \nu}(t) \, dt } \nonumber \\
 \mathfrak{m}_{n^\prime, \nu^\prime, n, \nu}(t) := & \bra{\nu^\prime} \, e^{t \mathcal{L}_A(\cdot)} \left( \bm{W}_\mu^{(k)}(t) \ket{n^\prime} \bra{n} \bm{W}_\mu^{(k)\,\dag}(t) \right)  \ket{\nu} \nonumber \\
 \mathfrak{r}_{n^\prime, \nu^\prime, n, \nu}(t) := & \bra{\nu^\prime} \, \left( \rhoa \otimes \trace_A \left[ \bm{W}_\mu^{(k)}(t) \ket{n^\prime} \bra{n} \bm{W}_\mu^{(k)\,\dag}(t)\right]  \right)  \ket{\nu} .
\end{align}}
Since $\rhoa \otimes \trace_A \left[ \bm{W}_\mu^{(k)}(t) \bm{X} \bm{W}_\mu^{(k)\,\dag}(t)\right]$ is a
completely-positive superoperator on $\bm{X}$, the $d^2 \times d^2$ Hermitian matrix 
$\mathfrak{r}_{n^\prime, \nu^\prime, n, \nu}(t)$ is non-negative,
and therefore, for any vector $z \in \mathbb{C}^{d^2}$:
{\small \begin{align}
& z^\dag \, R  \, z =0 \label{eq:intrtz1} \\ 
& \;\; \Longrightarrow
\;\; z^\dag \; \cdot \;  \mathfrak{r}(t) \; \cdot \;  z =0 \;\; \forall t \in \left[ 0, \, \frac{\pi}{c_{B^\dag}} \right]  \label{eq:intrtz2}  .
\end{align}}
Now take a $d^2$ vector $z$ such that (\ref{eq:intrtz1}) is satisfied. We then have:
{\small \begin{equation}
  \bra{\Psi} \bar{\mathcal{K}} (  \ket{\Phi} \bra{\Phi} ) \ket{\Psi} \ge 0 = 
 \int_{0}^{+\infty}{  \,\, \sum_{n^\prime, \nu^\prime, n, \nu} z_{n^\prime, \nu^\prime}^\star  \, \mathfrak{m}_{n^\prime, \nu^\prime, n, \nu}(t) \, z_{n^\prime, \nu^\prime} \,\,\,\, dt }   \nonumber ,
\end{equation}}
Since the propagator $e^{t \mathcal{L}_A(\cdot)}$ is a completely-positive mapping of the form
$e^{t \mathcal{L}_A(\cdot)}(\bm{X}) = \sum_{\theta}\bm{H}_\theta(t) \bm{X} \bm{H}(t)_\theta^\dag$ for some operators $\bm{H}_\theta(t)$, we then have \cite[Lemma 1]{azouitQST2017}:
{\small \begin{align}
& \int_{0}^{+\infty}{   \sum_{n^\prime, \nu^\prime, n, \nu} z_{n^\prime, \nu^\prime}^\star  \, \mathfrak{m}_{n^\prime, \nu^\prime, n, \nu}(t) \, z_{n^\prime, \nu^\prime} dt }  = \nonumber \\
&
\int_{0}^{+\infty}{   \sum_{n^\prime, \nu^\prime, n, \nu} z_{n^\prime, \nu^\prime}^\star   \bra{\nu^\prime} \bm{H}_\theta(t) \bm{W}_\mu^{(k)}(t) \ket{n^\prime} 
 \bra{n}  \bm{W}_\mu^{(k)\,\dag}(t) \bm{H}^{\dag}_\theta(t) \ket{\nu}  \, z_{n, \nu} dt }  = \nonumber \\
 &
 = \int_{0}^{+\infty}{    { \left| \sum_{n^\prime, \nu^\prime, n, \nu} \, 
 \bra{n}  \bm{W}_\mu^{(k)\,\dag}(t) \bm{H}^{\dag}_\theta(t) \ket{\nu}  \, z_{n, \nu} \right| }^2dt }   \label{eq:msumform} \\
& \; \ge \; 0  , \nonumber 
\end{align}}
and we thus conclude that $z^\dag R z = 0$ implies $z^\dag M z \ge 0$.
Assume now that $z^\dag R z = 0$ and $z^\dag M z =0$.
Inequality (\ref{eq:msumform}) then implies that for any $t \ge 0$
\begin{equation}
   \sum_{n^\prime, \nu^\prime, n, \nu} \, 
 \bra{n}  \bm{W}_\mu^{(k)\,\dag}(t) \bm{H}^{\dag}_\theta(t) \ket{\nu}  \, z_{n, \nu} = 0 \nonumber
\end{equation}
and therefore $M z =0$.
Then, by virtue of \cite[Lemma 2]{azouitQST2017}, we conclude that there exists 
$\bar{\tau} > 0$ such that $\bar{\mathcal{K}}_{\bar{\tau}}$ is completely positive.

\end{proof}
We are now going to prove trace-preservation of $\mathcal{K}(\rhos)$. Since $\trace [ \mathcal{K}_1(\rhos) ]=1$ and
$\trace [ \mathcal{K}_1(\rhos) ]=0$ for all $\rhos$,
what remains to prove is $\trace [ \mathcal{K}_2(\rhos) ]=0$ for all 
$\rhos$.
First, by a subsequent application of formulas (\ref{eq:Uequation2}), (\ref{eq:fhfjformula}), and 
(\ref{eq:singleFastFirstOrderSelectionOneMany}), we have:
{\small
\begin{align}
 & \sum_{k,k^\prime} \trace\left[ \bm{U}_{BB}^{(k,k^\prime)} \rhoakk  + \rhoakk  \bm{U}_{B^\dag B^\dag }^{(k,k^\prime)\, \dag} \right] = \nonumber \\
 & \;\; =  \sum_{k,k^\prime} \trace\left[ \bm{U}_{BB}^{(k^\prime,k)} \rhoakk  + \rhoakk  \bm{U}_{B^\dag B^\dag }^{(k,k^\prime)\, \dag} \right] = \nonumber \\
 & \;\; =  \sum_{k,k^\prime} \frac{1}{2 i c_{B^\dag} }\trace\left[ 
		    \bm{F}_2^{(k)} \rhoakk \bm{A}^{(k^\prime)\, \dag} + \bm{A}^{(k)\, \dag} \rhoakk \bm{F}_1^{(k^\prime) \, \dag}
 \right] = \nonumber \\
 & \;\; =  \sum_{k,k^\prime} \trace\left[ 
		    - \bm{F}_2^{(k)} \rhoakk \bm{F}_1^{(k^\prime)}
 \right] + \frac{1}{2 i c_{B^\dag}} \, \trace\left[ 
		    \check{\delta}_{k k^\prime} \bar{\mathcal{F}}_{2 1}(\rhoa)^{(k)}
 \right] = \nonumber \\
& \;\; = \sum_{k,k^\prime} \trace\left[ 
		    - \bm{F}_2^{(k)} \rhoakk \bm{F}_1^{(k^\prime)}
 \right] +  i \, \check{\delta}_{k k^\prime} \,  b^{(k)} . \label{eq:traceK2part1}
\end{align}}
Secondly, it is straightforward to prove from (\ref{eq:Uequation3}) and (\ref{eq:Uequation4})  that $ \sum_{k,k^\prime}\trace [\bm{U}_{B B^\dag}^{(k,k^\prime)} \rhoakk]$
and $\sum_{k,k^\prime} \trace [ \bm{U}_{B^\dag B}^{(k,k^\prime)} \rhoakk]$ are real. Indeed:
{\small
\begin{align}
& \sum_{k,k^\prime} \trace \left[ \bm{U}_{B B^\dag}^{(k,k^\prime)} \rhoakk \right] = 
\sum_{k} \trace \left[ - \bm{F}_1^{(k)} \rhoak \bm{F}_1^{(k)\, \dag}  - \bar{\tau} \bar{\mathcal{F}}_{11}( \rhoa )^{(k)} - \frac{\pi}{2} \right] \nonumber\\
& \;\;\;\; + \sum_{k > k^\prime} \trace \left[ - \bm{F}_1^{(k)} \rhoakk \bm{F}_1^{(k^\prime)\, \dag} - \bm{F}_1^{(k^\prime)} \rhoakk \bm{F}_1^{(k)\, \dag} \right]  = \nonumber \\
& \;\;\;\; = \sum_{k,k^\prime} \trace \left[  \rhoakk \bm{U}_{B B^\dag}^{(k,k^\prime)\, \dag } \right]   \\
& \sum_{k,k^\prime} \trace \left[ \bm{U}_{B^\dag B}^{(k,k^\prime)} \rhoakk \right] = 
\sum_{k} \trace \left[ - \bm{F}_2^{(k)} \rhoak \bm{F}_2^{(k)\, \dag}  - \bar{\tau} \bar{\mathcal{F}}_{22} ( \rhoa ) ^{(k)} - \frac{\pi}{2}\, |b^{(k)}|^2 \right] \nonumber\\
& \;\;\;\; + \sum_{k > k^\prime} \trace \left[ - \bm{F}_2^{(k)} \rhoakk \bm{F}_2^{(k^\prime)\, \dag} - \bm{F}_2^{(k^\prime)} \rhoakk \bm{F}_2^{(k)\, \dag} \right]  = \nonumber \\
& \;\;\;\; = \sum_{k,k^\prime} \trace \left[  \rhoakk \bm{U}_{B^\dag B}^{(k,k^\prime)\, \dag } \right]  \label{eq:traceK2part2} .
\end{align}}
Finally, thanks to observations (\ref{eq:traceK2part1}), (\ref{eq:traceK2part2}) and the definition of $f_1,f_2$ in (\ref{eq:K2defmore}),
we compute the trace of $\mathcal{K}_2$ from (\ref{eq:Kq2expanded}) as follows:
{\small \begin{align}
 \trace [ \mathcal{K}_2(\rhos )] = & \sum_{k,k^\prime} \Big\{
  \trace \left[  \bm{B}_j^\dag \bm{B}_h \rhos  \right] 
  \trace \left[  \bm{U}_{B_j^\dag B_h  }^{(k,k^\prime)} \rhoakk + \rhoakk \bm{U}_{B_j B_h^\dag}^{(k,k^\prime)\,\dag}
   + \bm{F}_h^{(k)} \rhoakk \bm{F}_j^{(k^\prime)\,\dag} \right]  \Big\}  \nonumber \\
& + \trace \left[ \bm{B}^\dag \rhos \bm{B} \right] \,
     \trace \left[ f_1 + \sum_{k}\left( \bar{\tau} \bar{\mathcal{F}}_{11}^{(k)} + \frac{\pi}{2} \right)  \right]  \nonumber \\
& + \trace \left[ \bm{B} \rhos \bm{B}^\dag \right] \,
     \trace \left[ f_2 + \sum_{k}\left( \bar{\tau} \bar{\mathcal{F}}_{22}^{(k)} + \frac{\pi}{2} \, | b^{(k)} |^2 \right)  \right]  \nonumber \\
& + \trace \left[ \bm{B} \rhos \bm{B} \right] \,
     \trace \left[ \sum_{k} -i b^{(k)}  \right]  + \trace \left[ \bm{B}^\dag \rhos \bm{B}^\dag \right] \,
     \trace \left[ \sum_{k} i b^{(k)}  \right]  \nonumber \\
= & \; 0 , \nonumber 
\end{align}}
for all $\rhos$.
We bring to the attention of the reader that the gauge choices in $\mathcal{K}_2$ have been carefully
designed so as to yield crucial quantum properties:
\begin{itemize}
 \item $\bar{\tau}$ is selected according to Claim \ref{claim:K2Qcptp} to yield complete positivity
 of $\mathcal{K}_2$;
 \item the trace of operators $\bm{U}_{B B}^{(k,k^\prime)}$, $\bm{U}_{B^\dag B^\dag}^{(k,k^\prime)}$ are selected
 according to (\ref{eq:Uequation2}) to yield cancellation of the $\bm{B} \bm{B} \rhos$ and
 $\bm{B}^\dag \bm{B}^\dag \rhos$ terms and their hermitian conjugates in the invariance equation, thus ensuring the Lindblad form of
 $\mathcal{L}_{s,2}$;
 \item constants $b^{(k)}$ are selected according to (\ref{eq:kmap2many}) to yield
 cancellation of the $\bm{B} \rhos \bm{B}$ and
 $\bm{B}^\dag \rhos \bm{B}^\dag$ terms in the invariance equation, thus ensuring the Lindblad form of
 $\mathcal{L}_{s,2}$;
 \item the definition of $f_1$, $f_2$ and the trace of operators $\bm{U}_{B^\dag B}^{(k,k^\prime)}$, $\bm{U}_{B B^\dag}^{(k,k^\prime)}$ are selected
 according to (\ref{eq:K2defmore})  and (\ref{eq:Uequation3})-(\ref{eq:Uequation4}) respectively in order 
 to cancel out the positive terms in the trace of $\mathcal{K}_2$ corresponding to quadratic terms,
 as shown in (\ref{eq:traceK2part2}),
 thus ensuring trace-preservation of $\mathcal{K}$.
\end{itemize}

\addtolength{\textheight}{-3cm}   



{\parskip = -2.8mm
\noindent
\bibliographystyle{plain}
\bibliography{QUANTIC}}

\end{document}